\documentclass[11pt]{article}

\usepackage{latexsym}
\usepackage{setspace}
\usepackage{cancel}

\usepackage{graphicx}
\usepackage{appendix}
\usepackage{amssymb}
\usepackage{dsfont}
\usepackage{amsmath}
\usepackage{amsthm}
\usepackage{cancel}
\usepackage{verbatim}
\usepackage{chngpage}
\usepackage{fullpage}

\begin{document}

\title{On the distributed evaluation\\ of recursive queries over graphs}
%\subtitle{Full Paper}

\author{St\'{e}phane Grumbach\thanks{INRIA-LIAMA, CASIA, PO Box 2728, Beijing 100190, PR China. Stephane.Grumbach@inria.fr}
\and Fang Wang\thanks{Lab of Computer Science, Institute of Software, Chinese Academy of Sciences, Beijing 100190. wangf@ios.ac.cn}\thanks{China Graduate School, Chinese Academy of Sciences, Beijing 100049, China}
\and Zhilin Wu\thanks{LIAMA, CASIA, PO Box 2728, Beijing 100190, PR China. zlwu@liama.ia.ac.cn}}

\newtheorem{definition}{Definition}
\newtheorem{theorem}{Theorem}
\newtheorem{proposition}{Proposition}
\newtheorem{corollary}[theorem]{Corollary}
\newtheorem{example}{Example}
\newtheorem{question}{Open Question}
\newtheorem{lemma}[theorem]{Lemma}
\newtheorem{Algo}{Algorithm}
\newtheorem{remark}{Remark}

%\newdef{theorem}{Theorem}
%\newtheorem{corollary}[theorem]{Corollary}
%\newtheorem{lemma}[theorem]{Lemma}
%\newdef{definition}{Definition}
%\newdef{example}{Example}
%\newdef{proposition}{Proposition}

\newcommand{\vs}{\vspace*{1cm}}
\newcommand{\vsfive}{\vspace*{5mm}}
\newcommand{\hsmthree}{\hspace*{-3mm}}
\newcommand{\hs}{\hspace*{1cm}}
\newcommand{\ctwo}{\cancel{\hspace*{2cm}}}

\newcommand{\lfo}{$FO_{loc}$}
\newcommand{\lfp}{$FP_{loc}$}

\date{}

\maketitle

\begin{abstract}
Logical formalisms such as first-order logic (FO) and fixpoint logic (FP) are
well suited to express in a declarative manner fundamental graph
functionalities required in distributed systems. We show that these
logics constitute good abstractions for programming distributed
systems as a whole, since they can be evaluated in a fully
distributed manner with reasonable complexity upper-bounds. We first
prove that FO and FP can be evaluated  with a polynomial number of
messages of logarithmic size. We then show that the (global) logical
formulas can be translated into rule programs describing the local
behavior of the nodes of the distributed system, which compute
equivalent results. Finally, we introduce local fragments of these
logics, which preserve as much as possible the locality of their
distributed computation, while offering a rich expressive power for
networking functionalities. We prove that they admit tighter
upper-bounds with bounded number of messages of bounded size.
Finally, we show that the semantics and the complexity of the local
fragments are preserved over locally consistent networks as well as
anonymous networks, thus showing the robustness of the proposed
local logical formalisms.
%\medskip
%{\bf Keywords:} logic, distributed computing, locality, complexity, protocols

\end{abstract}

\section{Introduction}

Logical formalisms have been widely used in different fields of
computer science to provide high-level programming abstractions. The
relational calculus used by Codd to describe data-centric
applications in an abstract way, is at the origin of the
technological and commercial success of relational database
management systems \cite{RG95}. Datalog, an extension of Horn clause
logic  with fixpoints, has been widely used to specify
functionalities involving recursion \cite{RD95}.

%Only few systems rely on logical languages to express applications
%distributed over networks. Distributed computing focuses mainly on
%the question of what can be computed in a distributed system. Unlike
%in sequential models, tiny differences on the systems might result
%in major differences in the complexity or even the computability of
%a problem.

The development of distributed applications over networks of devices
is generally a very tedious task, involving handling low level
system details. The lack of high-level programming abstraction has
been identified as one of the roadblocks for the deployment of
networks of cooperating objects \cite{WiSeNts06}.

Recently, the use of queries to define network applications has been
considered. Initially, the idea emerged in the field of sensor
networks. It was suggested to see the network as a database, and
interact with it through declarative queries. Several systems have
been developed, among which  Cougar \cite{FungSG02} and TinyDB
\cite{MaddenFHH05}, supporting SQL dialects.
Queries are processed in a centralized manner, leading to distributed
execution plans.

More recently, query languages were proposed as a mean to express
communication network problems such as routing protocols
\cite{LooHSR05} and declarative overlays \cite{LooCHMRS05}. This
approach, known as {\it declarative networking} is extremely
promising for it offers a high-level abstraction to program networks. It
was also shown how to use recursive queries to perform diagnosis of
asynchronous systems \cite{AbiteboulAHM05}, network monitoring
\cite{ReissH06}, as well as self-organization protocols
\cite{GrumbachLQ07}. Distributed query languages provide new
means to express complex network problems such as node discovery
\cite{AlonsoKSWW03}, route finding, path maintenance with quality of
service \cite{BejeranoBORS05}, topology discovery, including
physical topology \cite{BejeranoBGR03}, etc.

%In databases, query languages provide declarative means to query
%data independently of their physical representation. The database
%management system optimizes the declarative queries and produces an
%execution plans. Encapsulating the network in a database, and
%interacting with it through declarative queries offer considerable
%advantage:  (i) relying on very simple and short programs
%(essentially 2 orders of magnitude below their procedural
%equivalent), (ii) whose semantics can be defined formally.

However, there is a lack of systematic theoretical investigations of
query languages in the distributed setting, in particular on their
semantics, as well as the complexity of their distributed
computation. In the present paper,
we consider a distributed evaluation of classical query languages,
namely, first-order logic and fixpoint logic, which preserves
their classical semantics.

First-order logic and
fixpoint logic have been extensively investigated in the context of
database theory \cite{AbiteboulHV95} as well as finite model theory
\cite{EbbinghausFlum99}. Since the seminal paper of Fagin
\cite{fagin74}, showing that the class NP corresponds exactly to
problems which can be expressed in existential second-order logic,
many results have linked Turing complexity classes with logical
formalisms. Parallel complexity has also been considered for
first-order queries which can be evaluated in constant time over
circuits with arbitrary fan-in gates \cite{Immerman89}.

This raised our curiosity on the distributed potential of these
classical query languages to express the functionalities
of communication networks, which have to be computed in a
distributed manner over the network itself. If their computation can
be distributed efficiently, they can form the basis of a high level
abstraction for programming distributed systems as a whole.

We rely on the classical message passing
model \cite{AttiyaW04}. Nodes exchange messages with their neighbors
in the network. We consider four measures of complexity: (i) the
in-node computational complexity, rarely addressed in distributed
computing; (ii) the distributed time complexity; (iii) the message
size; and (iv) the per-node message complexity.
The behavior of the nodes is governed by an algorithm, the
distributed query engine, which is installed on each node, and evaluates
the queries by alternating local computation and exchange of queries and results
with the other nodes.

We first consider the distributed complexity of first-order logic and fixpoint logic with
inflationary semantics, which accumulates all the results of the different stages of the computation. Note that our result carry over for other formalisms such as least fixpoint.
We prove that the distributed complexity of first-order queries is in
$O(\log n)$ in-node time, $O(\Delta)$ distributed time ($\Delta$ is
the diameter of the network), messages of size $O(\log n)$, and a polynomial number
of messages per node. For fixpoint, a similar bound can be shown but
with a polynomial distributed time.

We then consider the translation of logical formulae that express
properties of graphs at a global level, into rule programs that
express the behavior of nodes at a local level, and compute the same
result. We introduce a rule language, $Netlog$, which extends Datalog,
with communication primitives, and is well suited to express
distributed applications, ranging from networking protocols to
distributed data management. $Netlog$ is supported by the Netquest
system, on which the examples of this paper have been implemented.
We prove that graph programs in Datalog$^\neg$
\cite{AbiteboulHV95} can be translated to $Netlog$ programs. Since it
is well known that first-order and fixpoint logics can be translated
in Datalog$^\neg$ \cite{EbbinghausFlum99}, it follows that global logical formulae
can be translated in behavioral programs in Netlog producing the same result.

Finally, we define local fragments of first-order and fixpoint logic,
respectively \lfo\/ and \lfp. These fragments provide a good
compromise in the trade-off between expressive power and efficiency
of the distributed evaluation. Important network functionalities
(e.g. spanning tree, on-demand routes etc.) can be defined easily in
\lfp. Meanwhile, its complexity is constant for all our measures,
but the distributed time which is linear in the diameter for \lfo\/
and in the size of the network for \lfp.

%Our techniques differ from the initial work on declarative
%networking \cite{LooCGGHMRRS06}, which focused on Datalog, and
%extension of classical optimization techiques to the distributed
%world, but resulting in a more complex not local language.

Our results shed light on the complexity of the distributed
evaluation of queries. Note that if the communication network is a
clique (unbounded degree), our machinery resembles Boolean circuits,
and we get constant distributed time, a result which resembles the
classical AC$^0$ bound \cite{Immerman89}.

We have restricted our
attention to bounded degree graphs and synchronous systems. Most of our algorithms carry over, or can be extended to
unrestricted graphs, and asynchronous computation, but not necessarily the complexity bounds.
Interestingly, the results for the local fragments carry over for
other classes of networks, such as locally consistent networks or
anonymous networks, thus showing the robustness of the languages
\lfo\/ and \lfp.

The paper is organized as follows. In the next section, we recall
the basics of first-order and fixpoint logics.
In Section~\ref{sec-machine}, the computation model
is presented. Section~\ref{sec-QE-FO} is devoted to distributed first-order
query execution, and Section~\ref{sec-QE-FP} to fixpoint query
execution. In Section~\ref{sec-netquest}, we introduce a behavioral language, Netlog, and show that FP formulae can be translated into equivalent Netlog programs.
In Section~\ref{sec-locfrag}, we consider the restriction to
the local fragments, and show that they can be evaluated over different types
of networks.

%%%%%%%%%%%%%%%%%%%%%%%%%%%%%%%%%%%%%%%%%%%%%%%%%%%
%%%%%%%%%%%%%%%%%%%%%%%%%%%%%%%%%%%%%%%%%%%%%%%%%%%
%%%%%%%%%%%%%%%%%%%%%%%%%%%%%%%%%%%%%%%%%%%%%%%%%%%
%%%%%%%%%%%%%%%%%%%%%%%%%%%%%%%%%%%%%%%%%%%%%%%%%%%

\section{Graph logics}\label{sec-ql}

We are interested in functions on graphs that represent the topology
of communication networks. We thus restrict our attention to {\it finite
connected bounded-degree undirected graphs}. Let $D$ be the bound on
the degree.

We assume the existence of an infinite ordered set of constants,
$U$, the {\it universe} of node Id's. A {\it graph}, ${\bf G}=(V,G)$, is
defined by a finite set of nodes $V\subset U$, and a set of edges
$G\subseteq V\times V$.

We express the functions on graphs as queries. A {\it query} of
arity $\ell$ is a computable mapping from finite graphs to finite
relations of arity $\ell$ over the domain of the input graph closed
under graph isomorphisms. A {\it Boolean query} is a query with Boolean output.

Logical languages have been widely used to define queries. A formula
$\varphi$ over signature  $G$ with $\ell$ free variables defines a
query mapping instances of finite graphs ${\bf G}$ to relations of
arity $\ell$ defined by: $A=\{(x_{1},\dots, x_{\ell})|{\bf G}\models\varphi(x_{1},\dots,x_{\ell})\}$.
We equivalently write ${\bf G},A\models\varphi$.

We denote by FO the set of queries definable using first-order
formulae. First-order queries  can be used in particular to check
locally forbidden configurations for instance. Their expressive
power is rather limited though.

Fixpoint logics on the other hand
allow to express fundamental network functionalities, such as those
involving paths.
If $\varphi(T;x_1,...,x_\ell)$ is a first-order formula with $\ell$
free variables over signature $\{G,T\}$, where $T$ is a new relation
symbol of arity $\ell$, called the {\it fixpoint relation}, then
$\mu(\varphi(T))$ denotes a fixpoint formula whose semantics is
defined inductively as the inflationary fixpoint $I$, of the
sequence:
%$I_{0}=\emptyset$; and $I_{i+1}=\varphi(I_i)\cup I_i, i\geq 0$,
\begin{eqnarray*}
I_{0}&=&\emptyset;\nonumber\\
\label{ind-fp} I_{i+1}&=&\varphi(I_i)\cup I_i, i\geq 0
\end{eqnarray*}
where $\varphi(I_i)$ denotes the result of the evaluation of
$\varphi(T)$ with $T$ interpreted by $I_i$. The $I_i$'s constitute
the \textit{stages} of the computation of the fixpoint. We write
$G,I \models \mu(\varphi(T))$, whenever $I$ is the fixpoint of the
formula $\varphi(T)$ as defined by the above induction.

It is well know \cite{AbiteboulHV95}  that on ordered domains, the class of graph
queries defined by inflationary fixpoint, denoted FP, captures
exactly all Ptime mappings, that is mappings that can be computed on
a Turing machine in time polynomial in the size of the graph.

The following examples illustrate the expressive power of FP for distributed applications.

The formula $\mu(\varphi(T)(x,h,d))$ for instance where the formula
$\varphi(T)(x,h,d)$ is defined by:
%
%\vspace*{-2mm}
\begin{eqnarray*}\label{OLSR}
(G(x,h)\wedge h=d) \vee (G(x,h) \wedge \exists z (T(h,z,d)\wedge x\neq z)
\wedge
\neg\exists u T(x,u,d))
\end{eqnarray*}
defines a table-based routing protocol (OLSR like) on the graph $G$,
where $h$ is the next hop from $x$ to destination $d$.

A spanning tree from a node $x$ satisfying $ReqNode(x)$
can be defined by a fixpoint
formula $\mu(\varphi(ST)(x,y))$, where the formula $\varphi(ST)(x,y)$ is defined by:

\noindent
$(G(x,y)\wedge ReqNode(x))  \vee $
 \begin{eqnarray*}
 \hspace{1cm} (\neg \exists x' ST(x',y) \wedge \exists w (ST(w,x)\wedge w\neq y)\wedge G(x,y) \wedge
    \forall w^\prime \forall x' (ST(w',x')\wedge G(x',y) \Rightarrow x'\geq x))
\end{eqnarray*}
% The relation $ST$ stores the children $y$ of a node $x$, on the node $x$.

Similarly, an On-Demand Routing protocol (AODV like), can be
defined by the fixpoint queries $\mu(\varphi(RouteReq)(x,y,d))$ and
$\mu(\psi(NextHop)(x,y,d))$, where $d$ is a constant and $\varphi(RouteReq)(x,y,d)$ is
defined by:

\noindent
$(G(x,y)\wedge ReqNode(x) \wedge dest(d)) \vee$
\begin{eqnarray*}
\left(\exists w (RouteReq(w,x,d)\wedge w\neq y)\wedge G(x,y)\wedge \right.
    \left.x \ne d \wedge \neg \exists w^\prime
    RouteReq(w^\prime,y,d)\right)
\end{eqnarray*}

\noindent
and $\psi(NextHop)(x,y,d)$ is defined by:

\noindent
\begin{eqnarray*}
(RouteReq(x,d,d)\wedge y=d) \vee\left(\exists z NextHop(y,z,d)\wedge RouteReq(x,y,d)\right)
\end{eqnarray*}
where a route request is first emitted by a node $x$ satisfying $ReqNode(x)$,
then a path defined by next hops from that node to destination $d$ is established
by backward computation on the route request.

%%%%%%%%%%%%%%%%%%%%%%%%%%%%%%%%%%%%%%%%%%%%%%%%%%%
%%%%%%%%%%%%%%%%%%%%%%%%%%%%%%%%%%%%%%%%%%%%%%%%%%%
%%%%%%%%%%%%%%%%%%%%%%%%%%%%%%%%%%%%%%%%%%%%%%%%%%%
%%%%%%%%%%%%%%%%%%%%%%%%%%%%%%%%%%%%%%%%%%%%%%%%%%%

\section{Distributed evaluation}\label{sec-machine}

We are interested in this paper in the distributed evaluation of queries.
We assume that each query to the network is posed by a {\it requesting node}
(the node satisfying the predicate $ReqNode(x)$ in the examples of the previous section).

The result of a query shall be distributed over all the nodes of the network.
In a query $Q(x_1,x_2,\cdots,x_\ell)$, one of the attributes $x_i$
denotes the {\it holding node}, written explicitly as $@x_i$,
that is the node which holds the results relative to $x_i$.
More precisely, the tuple
$\langle a_1,\cdots,a_{i-1},a,a_{i+1},\cdots,a_\ell \rangle$ is held by node $a$, such that
$Q(a_1,\cdots,a_{i-1},a,a_{i+1},\cdots,a_\ell)$ holds.
For simplicity, we will choose the first variable as holding attribute.

The results of fixpoint queries are thus distributed on holding
nodes. In the OLSR like example of the previous section,
each node shall hold its routing table as a result of the evaluation of the query.

The nodes of the network are equipped with a distributed query
engine to evaluate queries. It is a universal algorithm that
performs the distributed evaluation of any network functionality
expressed using queries. The computation relies on the message
passing model for distributed computing \cite{AttiyaW04}.

The {\it configuration} of a node is given by a state, an
in-buffer for incoming messages, an out-buffer for outgoing
messages, and some local data and metadata used for the computation.
We assume that the metadata on each node contain a unique
identifier, the upper bound on the size of the network, $n$, and the
diameter of the network, $\Delta$. We also assume that the local
data of each node includes all its neighbors with their identifiers.

We distinguish between {\it computation events}, performed in a
node, and {\it delivery events}, performed between nodes which
broadcast their messages to their neighbors. A sequence of computation
events followed by delivery events is called a {\it round} of the
distributed computation.

A {\it local execution} is a sequence of alternating
configurations and events occurring on one node. We assume that the
network is static, nodes are not moving, and that the communication
has no failure.

We assume that at the beginning of the computation of a query, all
the nodes are idle, in initial state, with their in-buffers, and
out-buffers empty. Note that, it is easy to extend the present
computational framework to a multithreaded computation with several
concurrent queries running in the network. The requesting node
broadcasts its query to its neighbors. The incoming messages in the
subsequent nodes trigger the start of their query engine computation.

The {\it evaluation} of a query {\it terminates} when the
out-buffers of all nodes are empty. The {\it result} is
distributed over the network in the memories of all nodes.
Note that alternative termination modes are also possible.

We consider four measures of the complexity of the distributed
computation:
\vspace*{-2mm}
\begin{itemize}
\item The {\it per-round in-node computational complexity},
IN-TIME/ROUND, is the maximal computational time of the in-node
computation in one round;

\item The {\it distributed time complexity}, DIST-TIME, is the maximum
number of rounds of any local execution of any node till the
termination;

\item The {\it message size}, MSG-SIZE, is the maximum number of bits
in messages;

\item The {\it per-node message complexity}, $\#$MSG/NODE, is
the maximum number of messages sent by any node till the termination
of the evaluation.
\end{itemize}
%
%\vspace*{-2mm}
There is a trade-off between the  in-node computation
and the communication. Our objective is to distribute the workload
in the network as evenly as possible, with a balanced amount of
computation and communication on each node. Clearly, centralized
computation can be carried on  by loading the topology of the
network on the requesting node, and performing the evaluation by
in-node computation. The centralized evaluation of FO and FP admits the following
complexity bounds.

%\vspace*{-2mm}

\begin{proposition}\label{the-central-comp-fo}
Let $G$ be a network of  diameter $\Delta$, with $n$ nodes. Let
$\varphi$ be a FO formula with $v$ variables. The complexity of the
centralized evaluation of the query $ \varphi$ on $G$ is given by:

\medskip
\begin{tabular}{cccc}
\hs IN-TIME/ROUND   &  DIST-TIME    & MSG-SIZE &  $\#$MSG/NODE \\
  $O(n^v \log n)$ & $O(\Delta)$ & $O(\log n)$ & $O(n)$
\end{tabular}

\noindent Suppose $\mu(\varphi(T)(x_1, \dots,x_\ell))$ is a FP
formula such that $T$ is a relational symbol of arity $\ell$,  and
it contains $v = \ell+k$ variables ($\ell$ free and $k$ bounded).
Then the complexity upper-bound of the centralized evaluation of the query
$\mu(\varphi(T)(x_1,\dots,x_{\ell}))$ on $G$ is the same as the
above complexity for FO formulae except for the \textit{IN-TIME/ROUND} which is in
$O(n^{\ell+v}\log n)$.
\end{proposition}

%\vspace*{-2mm}

\noindent Note that all nodes, but the requesting node, have $O(\log
n)$ per-round in-node complexity. The proof of this result follows
from classical results on data complexity of query languages
\cite{AbiteboulHV95}. In the sequel, we focus exclusively on
distributed query evaluation.

%%%%%%%%%%%%%%%%%%%%%%%%%%%%%%%%%%%%%%%%%%%%%%%%%%%
%%%%%%%%%%%%%%%%%%%%%%%%%%%%%%%%%%%%%%%%%%%%%%%%%%%
%%%%%%%%%%%%%%%%%%%%%%%%%%%%%%%%%%%%%%%%%%%%%%%%%%%
%%%%%%%%%%%%%%%%%%%%%%%%%%%%%%%%%%%%%%%%%%%%%%%%%%%

\section{Distributed complexity of FO}\label{sec-QE-FO}

In this section we show that the distributed evaluation of FO can be done with a polynomial number of messages but logarithmic in-node computation per round. The result relies on a naive distributed query engine for FO, $\mathcal{QE}_{FO}$, which works as follows.

The requesting node starts the computation by submitting a query.
The nodes broadcast Boolean answers to queries when they have them,
and otherwise queries they cannot answer, to their neighbors. Each
node reduces queries by instantiating variables. In
$\mathcal{QE}_{FO}$, nodes start instantiating from the leftmost
quantified variable, and from the rightmost free variable. The last
instantiated free variable therefore denotes the holding node of the
query, on which the corresponding tuples will be stored. The nodes
simplify the queries by removing all facts, or subformulae they can
fully evaluate.

Let $\varphi$ be a first-order formula with $\ell$ free variables.
The query engine handles the following message types: message
$\{?B\varphi\}$ for Boolean queries, message $\{?x_{1} \dots ?x_{i}
!a_{i+1} \dots !a_{\ell} \varphi\}$ for non-Boolean queries,
and message $\{! B \varphi\}$ for answers of Boolean queries.

Each node stores pairs ($query$,
$parentquery$), in a query table, associating the query being evaluated to the query
from which it derives. Nodes also store the Boolean answers $! B \varphi$ and non-Boolean
answers $\langle a_{1}\dots a_{\ell}\rangle$ to queries in an answer table.

We will see that the diameter $\Delta$ of the graph induces an
upper-bound on the response time of queries. The algorithm uses
clocks that are defined according to this upper-bound.
\textbf{Clocks} are associated to the evaluation of queries as well
as subqueries. After the time of a clock  associated to a query on a
node has elapsed, the value of the query can be determined by the node.
From now on, we assume that we are given a clock compliant with the
communication graph. The value of the clocks will be defined in
Definition~\ref{def-clock} below.

The main steps of the query engine work as follows. Note that we
assume for simplicity in the sequel that the system is synchronous.
This assumption can be relaxed easily in asynchronous systems without impact on the complexity by using spanning trees rather than the clocks.

\textbf{Initial Boolean query emission} For a Boolean query, the
requesting node, say $a$, broadcasts the query, $?B\varphi$, adds
($?B\varphi,nil$) into the query table, and sets a \textbf{clock}
for the answer. Meanwhile it instantiates the leftmost bounded
variable and produces a subquery. For an existentially quantified
formula $\exists x\psi$, if $\psi(a)$ is true then it is a
witness that $\exists x\psi$ is true.  For a universally
quantified formula $\forall x\psi$, if $\psi(a)$ is false then
it is a counterevidence and $\forall x\psi$ is false. If the node
doesn't have the answer to $\psi(a)$, it inserts $\psi(a)$
along with its parent query into the query table,
\textit{e.g.}$(?B\psi(a),?B\exists x\psi)$, broadcasts
$\psi(a)$ and also sets a \textbf{clock} for $\psi(a)$.
If no witness / counterevidence is received before the clock elapses, then
$\exists x\psi$ is false /  $\forall x\psi$ is true.
It then recursively handles $\psi(a)$ in the same way.

\textbf{Boolean query reception} Every node upon reception of a
Boolean query, $?B\varphi$, checks at first its query
table. If there is a record for this query, it does nothing.
Otherwise its behavior is similar to the Boolean query emission of
the requesting node, with the difference that it also broadcasts the
answer.

\textbf{Boolean answer reception} Every node receiving an answer to
a Boolean query, $! B \varphi$, checks its answer table. If there is
a record, it does nothing. Otherwise, it stores the answer,
checks the query table. If it is waiting for the answer, it then tries
to evaluate the parent query (if it has one), stores and broadcasts
its answer if it has; if it is not waiting for the answer, it
broadcasts $! B \varphi$.

\textbf{Initial non-Boolean query emission} The requesting node
submits and broadcasts the query $?x_{1} \dots ?x_\ell \varphi(x_1,\dots,x_{\ell})$.
It sets the clock, inserts $(?x_{1} \dots
?x_{\ell}\varphi(x_1,\dots,x_{\ell}), nil)$ into the query table,
instantiates the rightmost free variable to get the subquery, which is
$?x_{1}\dots ?x_{\ell-1}!a\varphi(x_1,\dots,x_{\ell-1},a)$, and broadcasts it. Meanwhile the subquery is inserted into the query table and handled further by the requesting node.
When all the free variables are instantiated, the Boolean query
$?B\varphi(a_1\dots a_\ell)$ is emitted and a record $(?B\varphi(a_1\dots a_\ell),$ $!a_1\dots !a_\ell\varphi(a_1\dots a_\ell))$ is inserted
in the query table of node $a_1$.

\textbf{Non-Boolean query reception} Every node
checks its query table when it receives a
%non-Boolean
query
$?x_{1}\dots?x_{i-1}!a_{i}\dots!a_{\ell}\varphi(x_1,\dots,x_{i-1},a_i,\dots,a_{\ell})$. If there is a record in the table, it does
nothing. Otherwise, it stores
$(?x_{1}\dots?x_{i-1}!a_{i}\dots!a_{\ell}\varphi(x_1,\dots,x_{i-1},$ $a_i,\dots,a_{\ell}),
nil)$ in the query table, its behavior is then similar to the
initial non-Boolean query emission with $i-1$ free variables.

\textbf{Distributed tuple answer collection} If the Boolean query
$?B\varphi(a_1,\dots,a_\ell)$ receives a positive answer to it, and there is a record $(?B\varphi(a_1\dots a_\ell), !a_1\dots !a_\ell\varphi(a_1\dots a_\ell))$ in the query table, $\langle a_1,\dots, a_\ell\rangle$ is stored in the answer table of
the current node which corresponds to the instantiation of the
leftmost free variable, that is the holding node for the answer.

\smallskip

\noindent We now turn to the  \textbf{clocks} which parameterize the
first-order query engines. The following theorem provides an
upper-bound on the distributed time complexity of the evaluation of
a formula.

%\vspace*{-2mm}

\begin{theorem}\label{the-clock}
For networks of diameter $\Delta$, the distributed time complexity
of the evaluation of a formula  with $w$ variables or constants by $\mathcal{QE}_{FO}$ is
bounded by $2\Delta w$.
\end{theorem}

%\vspace*{-2mm}

\begin{proof}
The proof is done by induction on the
number of variables and constants in the query $\psi$.

\textsc{Basis:} Assume $w=2$. There are three possibilities: two
constants, or two variables, or one constant and one variable in the
query $\psi$.
\begin{itemize}
\item If there are two constants, say $a$ and $b$, the query
$\psi$ is propagated to $a$ and gets the value of the atom
$G(a,b)$ which takes at most $\Delta$ rounds. Then the answer of
$\psi$ is sent back to the requesting node which takes at most
$\Delta$ rounds. The total time is at most $2\Delta$ rounds.
\item If there are one variable $x$ and one constant $a$ in $\psi$, the
query is propagated to every node at which the variable is
instantiated and we get the answers of $G(x,a)$, which takes $\Delta$
rounds.
\begin{enumerate}
\item[-] When the variable is free, the answer is stored in the local table of $x$.
\item[-] When the variable is bounded, the witness/counter
evidence of $\psi$ is sent back to the requesting node which takes
at most $\Delta$ rounds. Or if after $\Delta$ rounds, the requesting
node does not receive any sub-answer, it is sound to consider that  there
are no witnesses or counterevidences.
\end{enumerate}
So the total time is $2\Delta$ rounds in both cases.
\item If there are two variables then it takes $\Delta$ rounds to instantiate one
variable at every node (suppose the formula obtained is $\eta$) and then $\Delta$ rounds for the other variable (suppose the formula obtained
is $\xi$). Therefore $2\Delta$ in all.
\begin{enumerate}
\item[-] If both of the variables are free variables, if $\xi$ is
true, then the tuple is stored in the local table.
\item[-] If the first variable is free and the second one is bounded, then it takes
$\Delta$ rounds for the witness/ counter evidence (if there is one)
to get to the first instantiating node from the second one, if
$\eta$ is true, suppose $a$ is the instantiation of the free
variable, then the answer is stored in the local table.
\item[-] If both variables are bounded, then it takes $\Delta$ rounds for the answer to get to the
first instantiating node and then $\Delta$ to the requesting node.
\end{enumerate}
So the total time is $4\Delta$ rounds.
\end{itemize}
Therefore, for  $w=2$, the time is bounded by $2\Delta w$ rounds.

\textsc{Induction:} Suppose that when the sum of variables and
constants is $w$, \textit{e.g.} there are $l$ free variables, $k$
bounded variables, $c$ constants and $w=l+k+c$, the time is bounded
by $2\Delta w$ rounds. We prove the result for $w+1$
\begin{itemize}
\item when there are $c+1$ constants: there are $\Delta$ rounds
(at most) for the sub-query to get to the additional constant node
and $\Delta$ rounds for the answer to the sub-query getting back.
Therefore the total time is at most $2\Delta (w+ 1)$ rounds. \item
when there are $k+1$ bounded variables: \textit{w.l.o.g.} we assume
that the additional bounded variable is the leftmost bounded
variable, then $\Delta$ rounds are sufficient before instantiating
the second variable to instantiate the first variable, and $\Delta$
rounds for the answer getting to the first instantiating node from
the second one. Therefore the total time is at most $2\Delta
(w+1)$ rounds.
\item when there are $l+1$ free variables: it takes
$\Delta$ rounds for instantiating the additional free variable. So
the total time is $2\Delta w+ \Delta$.
\end{itemize}
Therefore, the distributed time time is bounded by
$2\Delta(w+1)$.
\end{proof}

We can now settle the values of the clocks in the query engine.

%\vspace*{-2mm}

\begin{definition}\label{def-clock}
The value of the clock in a network of diameter $\Delta$, for an FO
query  with $w$ variables or constants is $2\Delta w$.
\end{definition}

%\vspace*{-2mm}

The next result  shows the robustness of the algorithm: its
independence from the order in which messages are handled by the
query engine.

\begin{proposition}\label{the-mesind}
The distributed first-order query engine is insensitive to the order
of the incoming messages in a round.
\end{proposition}

%\vspace*{-3mm}

\begin{proof}
There are two fundamental steps in the algorithm of the query
engine: query propagation and result construction. During query
propagation, queries and subqueries arriving on one node have no
interaction. They generate entries in the query table. During result
construction, results of independent queries do not interfere, and
results of the same query are handled with a set semantics.
\end{proof}

We can now define the {\it distributed inference}.

%\vspace*{-3mm}
\begin{definition}
Let $G$ be a graph, $\psi$ a formula with $\ell$ free variables, and
$A$ a finite relation of arity $\ell$. We write $G, A
\vdash_{FO}\psi$ if and only if $A$ is  the union of all the answers
produced by the query engine $\mathcal{QE}_{FO}$ on all nodes, upon
request of $\psi$ from any node.
\end{definition}

%\vspace*{-3mm}

We next prove the soundness and completeness of the query engine.

%\vspace*{-3mm}

\begin{theorem}\label{the-theo-main-fo}
For any network $G$ of diameter at most $\Delta$, and any
first-order formula $\psi$, $G,A \models\psi$ if and only if
$G,A\vdash_{FO}\psi$.
\end{theorem}
%\vspace*{-3mm}
%
\begin{proof}
First observe that it is sufficient to prove the result for Boolean
formulae. Indeed, if there are $\ell$ free variables  in the query,
they get instantiated by all possible $n^{\ell}$ instantiation when
the query travels around the system of $n$ nodes, resulting in
$n^{\ell}$ Boolean first-order queries. The result of each query
(tuple of $\ell$ constants) is then stored at the key node if it
satisfies the Boolean query.

The result is also rather obvious for variable-free formulae.
Suppose that a query has $c$ ($c\geq 2$) constants and no variables.
The query is broadcasted to every node and once it successively
reaches nodes, it gets the Boolean value for the atoms containing
the corresponding constants, replaces the corresponding atoms by
their value and produces a new query which is broadcasted again. The
result is obtained when the query has reached (at most) $c-1$ of the
constants. Then the answer is sent back to the requesting node. The
total time required is at most $2\Delta (c-1)$. The \textbf{clock}
time being fixed at 2$\Delta c$ rounds, it is suffices to get the
result.

The rest of the proof is done by induction on the number of bounded
variables for Boolean formulae.

\textsc{Basis}: Assume the query has one bounded variable. Then it must has
at least one constants, so $c\geq 1$. First it is broadcasted by the
requesting node and the variable is instantiated by every node, thus
producing $n$ sub-queries with at most $c+1$ constants After the
sub-queries reach at most $c-1$ of the constants (note that one of
the constants stems from instantiating the variable and the
sub-queries gets it immediately at the instantiating node) and get
their answers, the witness for $\exists$ or the counterevidence for
$\forall$ is sent back to the requesting node which then produces
the final answer. If no witnesses/counterevidences are received
before the clock time elapses, a negative/positive answer is
produced by the requesting node.

\textsc{Induction}:  Assume that if the query has  $k$ ($k\geq 2$)
bounded
variables and $c$ constants, \textit{i.e.} the query is in the form:\\
$\psi_{k}=A_{1}x_{1}\dots A_{k}x_{k}\varphi(x_{1}\dots x_{k})$
(denoting $\exists$ or $\forall$ by $A$), then $G\models\psi_{k}$ if
and only if $G\vdash_{FO}\psi_{k}$.

We prove the result for the case when there are $k+1$ bounded
variables in the query
\[
\psi_{k+1}=A_{1}x_{1}\dots A_{k+1}x_{k+1}\varphi(x_{1}\dots x_{k+1})
\]
After the first variable has been instantiated at each node, the $n$
sub-queries of the form
\[
\psi'_{k}=A_{2}x_{2}\dots A_{k+1}x_{k+1}\varphi(x_{2}\dots x_{k+1})
\]
are queries with $k$ bounded variables and $c+1$ constants. They are
then further propagated by the instantiating node. By induction
assumption, $G \models\psi_{k}'$ if and only if
$G\vdash_{FO}\psi_{k}'$, so every node gets a sound answer to
$\psi'_{k}$. After one instantiating node gets the answer to
$\psi'_{k}$, it sends the answer to the requesting node. If it is
true and $A_{1}$ is $\exists$ then the requesting node takes it as a
witness and $\psi_{k+1}$ is true; if it is false and $A_{1}$ is
$\forall$ then the requesting node takes it as a counterevidence and
$\psi_{k+1}$ is false. If the requesting node does not receive any
witnesses/counterevidences until the clock time has elapsed, it
gives a negative/positive answer to $\psi_{k+1}$. Therefore $G
\models\psi_{k+1}$ if and only if $G\vdash_{FO}\psi_{k+1}$.
\end{proof}

We next consider the complexity of the distributed
evaluation. Theorem~\ref{the-comp-fo} is the fundamental result of
this section. It shows the potential for distributed evaluation of
first-order queries with logarithmic in-node time complexity,
distributed time linear in the diameter of the graph, and polynomial
amount of communication.

%\vspace*{-3mm}

\begin{theorem}\label{the-comp-fo}
Let $G$ be a graph of diameter $\Delta$, with $n$ nodes, and let
$\varphi$ be a first-order formula with $v$ variables. The
complexity of the distributed evaluation of the query $ \varphi$ on
$G$ by  $\mathcal{QE}_{FO}$ is given by:
\medskip
\begin{tabular}{cccc}
\hs IN-TIME/ROUND   &  DIST-TIME    & MSG-SIZE &  $\#$MSG/NODE \\
  $O(\log n)$ & $O(\Delta)$ & $O(\log n)$ & $O(n^{v+1})$
\end{tabular}
\end{theorem}

\begin{proof}
(sketch)\\
We assume that $\varphi$ has $\ell$ free variables, $k$ bounded
variables and $c$ constants. So $v=\ell+k$. Let $w=v+c$.

\noindent IN-TIME/ROUND\\
We consider the complexity in the size of the graph.  The query is
partially evaluated on the local data (identifiers of neighbors) of
$O(\log n)$ size. It is rewritten in a systematic fashion into
sub-queries by instantiating variables. Both operations can be
performed in $O(\log n)$ time. The searching on the query table and
answer table (both of size $O(n^v)$) can be done in $O(\log n)$ time
as well by binary searching.

\noindent DTIME\\
As shown in Theorem~\ref{the-clock}, the distributed time for a
query is $2\Delta w$, so the time complexity is in $O(\Delta)$.

\noindent MSG-SIZE\\
It is evident that MSG-SIZE is $O(\log n)$.

\noindent $\#$MSG/NODE\\
During the distributed evaluation of queries, new queries can be
generated by instantiating free and bounded variables. The total
number of queries generated during the distributed evaluation is
$O(\sum_{i=1}^{v}n^i)$, which is $O(n^{v+1})$. So the number of
queries and answers received by each node is $O(n^{v+1})$.
Therefore, the number of messages sent by each node is $O(n^{v+1})$.
\end{proof}

Note that the first-order query engine relies on a naive evaluation
of queries. It can be optimized by taking advantage of the patterns
in the query to limit the propagation of subqueries, but this does
not affect the global complexity upper bounds.

\section{Distributed complexity of FP}\label{sec-QE-FP}

We next consider the complexity upper bounds for FP. It relies on a
query engine which is defined as follows. Note that we first assume
that the system is synchronous and we discuss asynchronous
systems at the end of the present section.

\noindent
\textbf{Query engine for $FP$}, $\mathcal{QE}_{FP}$.
At first, the requesting node broadcasts $\mu(\varphi(T)(
x_1,\dots,x_{\ell}))$ (where $T$ is a relational symbol of arity $\ell$).
It takes $\Delta$ rounds for all nodes to receive the query. In
order to coordinate the computation of the stages of the fixpoint on
different nodes, a hop counter $c$ is broadcasted together with the
query $\mu(\varphi(T)(x_1,\dots,x_{\ell}))$, and a clock $\sigma$ is
set for each node. Initially, the requesting node sets $\sigma =
\Delta$, and broadcasts
$(\mu(\varphi(T)(x_1,\dots,x_{\ell})),\Delta-1)$ to its neighbors.
Each node receiving messages of the form
$(\mu(\varphi(T)(x_1,\dots,x_{\ell})),c)$ sets $\sigma=c$ and
propagates the formula $(\mu(\varphi(T)(x_1,\dots,x_{\ell})),c-1)$ to its
neighbors, unless $c=0$ or $\sigma$ has been set before.

When the clock $\sigma$ expires, each node $a$ sets a local table
for $T$ and performs the recursion on $\mu(\varphi(T))$ by iterating
the use of the first-order query engine $\mathcal{QE}_{FO}$ on the
query $\varphi(T)$ as follows:

\vspace*{-2mm}
\begin{itemize}
\item $a$ sets a clock $\tau=2\Delta w$ (where $w$ is the number of
variables or constants in $\varphi(T)$), evaluates the query $?x_1\dots
?x_{\ell-1}!a \varphi(T)(x_1,\dots,x_{\ell-1},a)$ using
$\mathcal{QE}_{FO}$, which takes time $2\Delta w$.

\item If $a$ receives a query $?x_1!a_2\dots !a_{\ell}\varphi(T)(x_1,a_2,\dots,a_{\ell})$ before $\tau$ expires,
$x_1$ is instantiated by $a$ to get the
subquery $!a!a_2\dots !a_{\ell}\varphi(T)(a,a_2,\dots,a_{\ell})$,
and the evaluation of the Boolean query
$?B\varphi(T)(a,a_2,\dots,a_{\ell})$ starts. If $a$ gets a positive answer
to that Boolean query,
%$?B\varphi(T)(a,a_2,\dots,a_{\ell})$,
it stores $\langle a,a_2,\dots,a_{\ell}\rangle$ in a temporary buffer.

\item When the clock $\tau$ expires, node $a$ updates the local table for $T$ and
sets another clock $\eta=\Delta$. If some new tuples $\langle
a,a_2,\dots,a_{\ell}\rangle$ have been produced, $a$ broadcasts an
informing message to its neighbors, which will be propagated further
to all the nodes in the network to inform them that the computation has not reached a fixpoint yet.

\item If some new tuples have been produced in $a$ or $a$ has received some informing
messages when the clock $\eta$ expires, it resets $\tau=2\Delta w$ and starts the next
iteration, otherwise the evaluation terminates. \qed
\end{itemize}

\begin{definition}
Let $\mu(\varphi(T))$ be a fixpoint formula,  $G,
I\vdash_{FP}\mu(\varphi(T))$ if and only if upon request of
$\mu(\varphi(T))$ from any node $a$, the query engine
$\mathcal{QE}_{FP}$ produces answer $I$ distributed in the network.
\end{definition}

%\vspace*{-5mm}

As for FO, we show that the query engine is sound and complete.

\begin{theorem}\label{the-theo-fp}
For a network $G$ and   $\mu(\varphi(T))$ a  fixpoint formula,
$G,I\models\mu(\varphi(T))$ if and only if $G,I\vdash_{FP}\mu(\varphi(T))$.
\end{theorem}

%\vspace*{-2mm}

\noindent The proof of Theorem~\ref{the-theo-fp} follows easily from
Theorem~\ref{the-theo-main-fo}.

%\vspace*{-3mm}

\begin{theorem}\label{the-comp-fp}
Let $G$ be a graph of diameter $\Delta$, with $n$ nodes, $T$ a
relation symbol of arity $\ell$, and
$\mu(\varphi(T)(x_1,\dots,x_{\ell}))$ be a FP formula with $v =
\ell+k$ (first-order) variables ($\ell$ free and $k$ bounded). The
complexity of the distributed evaluation of the query
$\mu(\varphi(T))$ by $\mathcal{QE}_{FP}$ on $G$ is given by:

\medskip
\begin{tabular}{cccc}
\hs IN-TIME/ROUND   &  DIST-TIME    & MSG-SIZE &  $\#$MSG/NODE \\
  $O(\log n)$ & $O(n^{\ell}\Delta)$ & $O(\log n)$ & $O(n^{\ell+v+1})$
\end{tabular}

\end{theorem}

%\vspace*{-2mm}

\begin{proof}
Let $w$ be the total number of variables and constants in
$\varphi(T)(x_1, \dots, x_{\ell})$.

Messages
$(\mu(\varphi(T)(x_1,\dots,x_{\ell})),hop)$ are transferred in the
network, before the clock $\sigma$ expires, which takes $O(\Delta)$ round and $O(1)$ messages for each
node.

Queries $?x_1\dots ?x_{\ell-1}!a\varphi(T)(x_1,\cdots,x_{\ell-1},a)$ are
evaluated after the clock $\sigma$ expires, before $\tau$ expires. $O(n^{v})$ messages are sent by each node for each
such query (there are at most $v-1$ variables in
$\varphi(T)(x_1,\cdots,x_{\ell-1},a)$) by Theorem~\ref{the-comp-fo}.
Since there are $n$ such $?x_1\dots
?x_{\ell-1}!a\varphi(T)(x_1,\cdots,x_{\ell-1},a)$ queries, the total
number of messages sent by each node is $O(n^{v+1})$.

When $\tau$ expires, each node sets a clock $\eta=\Delta$, and
broadcasts informing messages to its neighbors if some new tuples
are produced. Each node receives the informing message will
broadcast it to its neighbors unless it has done that before. Each
node sends $O(1)$ informing messages before $\eta$ expires.

When $\eta$ expires, if a node has produced some new tuples or
received some informing messages during the previous iteration, it
starts the next iteration.

So before the evaluation terminates, in each iterating period $2\Delta w
+\Delta$ after the expiration of $\sigma$, at least one new tuple in
$T$ is produced in some node, thus there are at most $n^{\ell}$ such
periods before the termination of the evaluation since there are at
most $n^{\ell}$ tuples in $T$.

Consequently the total time of the evaluation is in $\Delta+
n^{\ell}(2\Delta w + \Delta)=O(n^{\ell}\Delta)$.

Because in each such period, $O(n^{v+1})$ messages are sent by each
node, so the total number of messages sent by each node before the
termination of the evaluation is $O(n^{\ell+v+1})$.
\end{proof}

Although the complexity upper-bound for DIST-TIME and $\#$MSG/NODE
is polynomial, the exponent relates to the number of variables. For
most networking functionalities, this number is small, and the
dependencies between the variables, might even lower it.

The algorithm $\mathcal{QE}_{FP}$ above can be adapted to an asynchronous system
by using a breath-first-search (BFS) spanning tree (with the
requesting node as the root), without impact on the complexity
bounds. If an arbitrary spanning tree, not necessarily a BFS tree,
is used, then the complexity bounds does not change, except the
distributed time, which becomes $O(n^{\ell+1})$.

\smallskip

Note that with $\mathcal{QE}_{FP}$, nodes are coordinated to compute
every stage of the fixpoint simultaneously by using the clock
$2\Delta w$, which is critical for preserving the centralized
semantics of $FP$ formulae. However if $\varphi$ is monotone on $T$,
the centralized semantics of the fixpoint is preserved no matter whether the
stages are computed simultaneously or not.
Similar results can be shown for alternative definitions of the fixpoint logic, such as Least Fixpoint.

%%%%%%%%%%%%%%%%%%%%%%%%%%%%%%%%%%%%%%%%%%%%%%%
%%%%%%%%%%%%%%%%%%%%%%%%%%%%%%%%%%%%%%%%%%%%%%%
%%%%%%%%%%%%%%%%%%%%%%%%%%%%%%%%%%%%%%%%%%%%%%%

\section{In-node behavioral compilation}\label{sec-netquest}

In this section, we see how to transform FO and FP formulae, which express queries at the global level of abstraction of the graph, to equivalent rule programs that model the behavior of nodes. We first introduce the $Netlog$ language.

A $Netlog$ program is a finite set of rules of the form:
\[
(\uparrow)~\gamma_0 :- \gamma_1;\dots;\gamma_l.
\]
where $l\geq 0$.
The {\it head} of the rule $\gamma_0$ is an atomic  first-order formula.
The {\it body}, $\gamma_1;\dots;\gamma_l$ is constituted of literals, i.e., atomic ($R(\overrightarrow x)$) or negated atomic ($\neg R(\overrightarrow x)$) formulae. Each atomic formula $\gamma_i$ has a {\it holding variable}, which is written explicitly as $@x$ and specifies the node on which the evaluation is performed. The {\it communication construct}, $\uparrow$, is added before the head if the result is to be pushed to neighbors.

In the sequel we denote the head of a rule $r$ as $head_r$ and the body as $body_r$ and denote the holding variable of a formula $\gamma_i$ as $hv_{\gamma_i}$.
The relations occurring in the head of the rules are called {\it intentional relations}.

Some localization restrictions are imposed on the rules to ensure the effectiveness of the distributed evaluation.
%\vspace*{-2mm}
\begin{enumerate}
\item[(i)] All literals in the body have the same holding variable;
%\vspace*{-4mm}
\item[(ii)] the head is not pushed (by $\uparrow$) if the holding variable of the head is the holding variable of the body;
%\vspace*{-2mm}
\item[(iii)] if the head is pushed (by $\uparrow$), assuming the holding variable of the head is $x$ and the holding variable of the body is $y$, then $G(@y,x)$ is in the body.\end{enumerate}
%\vspace*{-2mm}

A Netlog program is running on each node of the network concurrently.
All the rules are applied simultaneously on a node. The holding variable of literals in the body is instantiated by the node ID itself. Facts deduced are stored on the node if the rule is not modified by $\uparrow$. Otherwise, they are sent to nodes interpreting the holding variable of the head.

On each node,  (i) phases of executions of the rules on the node and (ii) phases of communication with other nodes are alternating till no new facts are deduced on each node.
The global semantics is defined as the union of the facts obtained on each node.

For a graph ${\bf G}=(V,G)$, an instance $I$ such that $I=\bigcup\limits_{v\in V} I_v$ where $I_v$ is the fragment of $I$ stored on node $v$, a rule:
\[r: Q(\overrightarrow x):-R_1(\overrightarrow {y_1});\dots;R_m(\overrightarrow {y_m});\neg R_{m+1}(\overrightarrow {y_{m+1}});\dots;\neg R_l(\overrightarrow {y_l}).\]

and an instantiation $\sigma$ of the variables occurring in $r$,
\[(I,\sigma)\models_{\bf G} R_1(\overrightarrow {y_1});\dots;R_m(\overrightarrow {y_m});\neg R_{m+1}(\overrightarrow {y_{m+1}});\dots;\neg R_l(\overrightarrow {y_l})\]
if and only if
\[
R_i(\sigma(\overrightarrow {y_i}))\left\{\begin{array}{ll}\in I_{\sigma(y)}\cup G, \mbox{ for }i\in[1,m]\\\notin I_{\sigma(y)}\cup G, \mbox{ for }i\in[m+1,l]\end{array}\right.\]
where $y$ is the holding variable of $body_r$.

We define the {\it immediate consequence operator} of a $Netlog$ program $P$ as a mapping from an instance $I$ to an instance:
\[
\Psi_{P,{\bf G}}(I)=\bigcup\limits_{v\in V}\left\{Q(\overrightarrow u)\left|\begin{array}{ll} \exists r\in P: Q(\overrightarrow x):-body_r\\\exists\sigma s.t. (I,\sigma)\models_{\bf G} body_r;\\\overrightarrow u=\sigma(\overrightarrow x);\sigma(hv_{Q(\overrightarrow x)})=v.\end{array}\right.\right\}
\]

The {\it computation of a $Netlog$ program $P$ on a graph ${\bf G}$} is given by the following sequence:
\begin{eqnarray*}
I_0&=&\emptyset;\\
I_{i+1}&=&\Psi_{P,{\bf G}}(I_i),i\geq 0
\end{eqnarray*}
The computation of $P$ on ${\bf G}$ {\it terminates} if the sequence $(I_i)_{i\geq 0}$ converges to a fixpoint. If the computation of $P$ on ${\bf G}$ terminates, we define $P({\bf G})$ to be the least fixpoint obtained by the computation sequence $(I_i)_{i\geq 0}$.

\bigskip

Before we see how  FO or FP formulae can be rewritten into $Netlog$ programs, let us first illustrate the technique on the  examples of Section~\ref{sec-ql}.

\begin{example}
The following program computes the OLSR like table-based routing protocol
as defined in Section~\ref{sec-ql}:

\begin{eqnarray*}
 T(@x,d,d)&:-&G(@x,d).\\
%ExistT(x,d)&:-&G(x,d).\\
T(@x,h,d)&:-& \neg existT(@x,d); G(@x,h);askT(@x,h,d).\\
%ExistT(x,d)&:-& \neg ExistT(x,d); G(x,h); T'(x,h,z,d).\\
existT(@x,d)&:-&T(@x,u,d).\\
 \uparrow askT(@x,h,d)&:-&T(@h,z,d); G(@h,x);x\neq z.\\
T(@x,d,d)&:-& T(@x,d,d).
\end{eqnarray*}
New predicates ($askT$) are introduced to store partial results that are computed on some nodes,
and used by other nodes to which they have been forwarded. The last rule ensures the inflationary behavior (accumulation of results).
\end{example}
%
%
%\vspace*{-6mm}

\begin{example}
The following program computes spanning trees as defined in Section~\ref{sec-ql}.
Several new predicates are introduced to reduce the complexity of the formula ($delay,rej$) and to ensure the transfer of data between the nodes involved in the computation ($askST$).
\begin{eqnarray*}
%Synchronous\\
\uparrow ST(x,@y)&:-&G(@x,y); ReqNode(@x).\\
%
%ST(x,y)&:-&ansST(x,y).\\
%
ST(x,@y)&:-&\neg existST(@y);delay(x,@y); \neg rej(x,@y).\\
\uparrow askST(x,@y)&:-&ST(w,@x);G(@x,y);w\neq y.\\
existST(@y)&:-&ST(x,@y).\\
rej(x',@y)&:-&askST(x,@y);askST(x',@y); x'\geq x.\\
delay(x,@y)&:-&askST(x,@y).\\
ST(x,@y)&:-&ST(x,@y).
\end{eqnarray*}

\end{example}
%
%%XXX:SG A figure would be nice to illustrate the distributed computation of the ST.
%
%\vspace*{-5mm}

\begin{example}
The following program computes the AODV like on-demand routing protocol
as defined in Section~\ref{sec-ql}.

\begin{eqnarray*}
%Synchronous\\
\uparrow RouteReq(x,@y,d)&:-&G(@x,y); ReqNode(@x);dest(d).\\
RouteReq(x,@y,d)&:-&askRouteReq(x,@y,d); \neg existRR(@y,d).\\
\uparrow askRouteReq(x,@y,d)&:-& RouteReq(w,@x,d); G(@x,y);x\neq d;w\neq y.\\
existRR(@y,d)&:-&RouteReq(w',@y,d).\\
\uparrow Nexthop(@x,d,d)&:-&RouteReq(x,@d,d); G(@d,x).\\
\uparrow Nexthop(@x,y,d)&:-& RouteReq(x,@y,d); Nexthop(@y,z,d);G(@y,x).\\
RouteReq(x,@y,d)&:-&RouteReq(x,@y,d).\\
Nexthop(@x,d,d)&:-&Nexthop(@x,d,d).\\
\end{eqnarray*}

\end{example}

We now consider the general translation of FO and FP formulae to $Netlog$ programs.
It has been shown in \cite{EbbinghausFlum99} that FP is equivalent to $Datalog^\neg$ both with inflationary semantics. Moreover, both FO and FP formulae can be translated effectively to $Datalog^\neg$ programs. We therefore consider the translation of $Datalog^\neg$ programs into equivalent $Netlog$ programs. The main difficulty relies in the distribution of the computation.

The syntax and semantics of $Datalog^\neg$ is similar to the one of $Netlog$, but without the communication primitives. Indeed, unlike $Netlog$, a program in $Datalog^\neg$, is processed in a centralized manner. The {\it computation of a $Datalog^\neg$ program $P$ on a graph ${\bf G}$} is given by the following sequence:
\begin{eqnarray*}
I_0&=&\emptyset;\\
I_{i+i}&=&\Psi_{P,{\bf G}}(I_i)\cup I_i,i\geq 0
\end{eqnarray*} where $\Psi_{P,{\bf G}}(I_i)$ is defined in a similar way as for $Netlog$.

The following algorithm rewrites a $Datalog^\neg$ program
$\mathcal{P}_{DL}$ into a $Netlog$ program $\mathcal{P}_{NL}$. To
synchronize stages of the recursion, there is a fact ``$start(a)$''
stored on each node $a$ at the beginning of the computation which
triggers a clock used to coordinate stages.

In the sequel we do not distinguish between $G(@x,y)$ and $G(@y,x)$. % The consistent holding variable of the literals in the body of a rule $r$ is denoted as $h_r$.

\noindent{\bf Rewriting Algorithm:}

The algorithm rewrites the input program step by step.

\noindent{\bf Step 1: Distributing Data}

Input: $\mathcal{P}_{DL}$. Output: $\mathcal{P}_1$.

Algorithm $Localize(\mathcal{P}_{DL})$ chooses one variable as the holding variable for each relation in $\mathcal{P}_{DL}$. $\mathcal{P}_1$ is obtained by marking the holding variable of each literal in $\mathcal{P}_{DL}$.

The Rewriting Algorithm supports different assignment of holding variables. For simplicity, we assume the left most variable of each relation is chosen as holding variable.  For lack of space, we do not address the associated optimization problem.

\noindent{\bf Step 2: Distributing Computation}

Input: $\mathcal{P}_1$. Output: $<P_2,\kappa>$

Let $\Delta$ be the diameter of {\bf G}.

For each rule $r\in \mathcal{P}_1$, assume
%\vspace*{-2mm}
\begin{itemize}
\item $hv_{head_r}$ is the holding variable of $head_r$,
%\vspace*{-1mm}
\item $h_r:=hv_{head_r}$, and
%\vspace*{-1mm}
\item $CN_r:=\{h_r\}$.
%\vspace*{-2mm}
\end{itemize}

$Rewrite(r,h_r,CN_r)$ recursively rewrites the rule $r$ into several rules until the output rules satisfy the localization restriction (i). $body_r$ is divided into several parts: the local part that can be evaluated locally and the non-local part that cannot be evaluated locally. $h_r$ is the holding variable of the literals in the local part. The non-local part is partitioned into several disconnected parts which share no variables except the variables in $CN_r$ and are evaluated by additional rules $r_i$ on different nodes in parallel. The deduced facts of $r_i$ are pushed to the node where the rule $r$ is evaluated. Meanwhile, it calculates the number of rounds $\kappa_r$ for evaluating $r$.

\noindent${\bf Rewrite(r,h_r,CN_r)}: \mbox{ output }<T_r,\kappa_r>$

Begin

Assume \[r:\gamma :- \gamma_1; \dots; \gamma_l.\] where $l\geq 1$.

Let $S=\{\gamma_1, \dots, \gamma_l\}$, $S'=\{\gamma_i|\gamma_i\in S\mbox{ and } hv_{\gamma_i}=h_r\}$, so that $S'$ contains all the literals in $body_r$ whose holding variable is the same as the one of the head, $h_r$.
\hspace*{-1cm}
\begin{enumerate}
\item[-] If $S'=S$, then $T_r:=\{r\}$, and $\kappa_r:=1$.
\item[-] If $S'\neq S$,

Begin

Let $S'':=S-S'$, so that $S''$ contains all the literals in $body_r$ whose holding variables are not $h_r$.

For $\gamma_j,\gamma_k\in S''$, let $\gamma_j \approx \gamma_k$ if $\gamma_j$ and $\gamma_k$ have some common variables besides the variables in $CN_r$. Assume $\{S''_{1},\dots,S''_{n}\}~(n\geq 1)$ is a partition of $S''$ in minimal subsets closed under $\approx$, so that the literals in $S''$ are divided into disconnected ''subgraph'' components.

For each $S''_i,i\in[1,n]$, let
\[T_i:=\{hv_{\gamma_{iw}}\vert\gamma_{iw}\in S''_i\mbox{ and } G(@h_r,hv_{\gamma_{iw}})\in S'\}.\]
so that $T_i$ contains the variables which are the holding variable of one literal in $S''_i$ and are also a neighbor of $h_r$.
\begin{enumerate}
\item[-] If $T_i\neq\emptyset$, which means the non-local part $S''_i$ is connected with the local part $S'$. Choose one variable $hv_{\gamma_{iu}}$ from $T_i$. Let $S''_i:=S''_i\cup \{G(@hv_{\gamma_{iu}},h_r)\}$. Let $h_{r_i}:=hv_{\gamma_{iu}}$. Let $CN_{r_i}:=CN_r\cup \{h_{r_i}\}$. Let $d_{r_i}:=1$. Assume $S''_i=\{\gamma_{i,1},\dots,\gamma_{i,m_i}\}$. Let
\[r_{i}: Q_{i}(\overrightarrow {y_{i}}) :- \gamma_{i,1};\dots;\gamma_{i,m_i}.\]
where $Q_i$ is a new relation name and $\overrightarrow {y_j}$ contains all the variables occurring both in $S''_i$ and in either $S'$ or $head_r$, that is in $var(S''_i)\cap(var(S')\cup var(head_r))$, with $h_r$ as holding variable.
\item[-] If $T_i=\emptyset$, then the non-local part $S''_i$ is disconnected from the local part $S'$. Choose one literal $\gamma_{it}\in S''_i$. Assume $y$ is a variable not occurring in $r$, let $S''_i:=S''_i\cup \{y=hv_{\gamma_{it}}\}$. Let $h_{r_i}:=hv_{\gamma_{it}}$. Let $CN_{r_i}:=CN_r\cup \{h_{r_i}\}$. Let $d_{r_i}:=1+\Delta$. Assume $S''_i=\{\gamma_{i,1},\dots,\gamma_{i,m_i}\}$. Let
\[r_{i}: Q_{i}(\overrightarrow {y_{i}}) :- \gamma_{i,1};\dots;\gamma_{i,m_i}.\]
where $Q_i$ is a new relation name and $\overrightarrow {y_i}$ contains all the variables occurring both in $S''_i$ and in either $S'$ or $head_r$, that is in $var(S''_i)\cap(var(S')\cup var(head_r))$, with $y$ as holding variable. Moreover, let
\[r'_{i}: Q_{i}(@x\dots):-Q_{i}(@y\dots);G(@y,x).\]
\end{enumerate}
Assume $S'=\{\gamma'_{1},\dots,\gamma'_k\}~(k\geq 0)$, let
\[r': \gamma :- \gamma'_{1};\dots;\gamma'_k; Q_{1}(\overrightarrow {y_{1}});\dots;Q_{n}(\overrightarrow {y_{n}}).\]
%where $Q_{1}(\overrightarrow {y_{1}}),\dots,Q_{n}(\overrightarrow {y_{n}})$ will be defined below.
$Q_i(\overrightarrow {y_i})$, $i\in[1,n]$, is called sub-query.

%Let $\kappa_{r_i}=max\{k_i|i\in[1,n]\}$.
Assume $<T_{r_i},\kappa_{r_i}>=Rewrite(r_i,h_{r_i},CN_{r_i})$, let
%\vspace*{-4mm}
\begin{itemize}
\item $T_r:=\{r'\}\cup\bigcup\limits_{i\in[1,n]}(\{r'_i\}\cup T_{r_i})$, and
%\vspace*{-1mm}
\item $\kappa_r:=max\{\kappa_{r_i}+d_{r_i}|i\in[1,n]\}$,
%\vspace*{-2mm}
\end{itemize}

End
\end{enumerate}
End

Finally, let
%\vspace*{-4mm}
\begin{itemize}
\item
$P_2:= \bigcup\limits_{r\in \mathcal{P}_1}T_r$, and
%\vspace*{-1mm}
\item $\kappa:=max\{\Delta,max\{\kappa_r|r\in \mathcal{P}_1\}\}.$
%\vspace*{-2mm}
\end{itemize}

\noindent{\bf Step 3: Communication}

Input: $<\mathcal{P}_2,\kappa>$. Output: $<\mathcal{P}_3,\kappa>$.

$\mathcal{P}_3$ is obtained by adding $\uparrow$ in the head of each rule $r$ where $r\in\mathcal{P}_2$ with the holding variable of the head different from the holding variable of the body. So that rules in $\mathcal{P}_3$ satisfy the localization restriction (ii) and (iii).

\noindent{\bf Step 4: Stage coordination with clocks}

Input: $<\mathcal{P}_3,\kappa>$. Output: $\mathcal{P}_4$.

The rules in $\mathcal{P}_3$ are modified as follows:
%\vspace*{-3mm}
\begin{itemize}
\item[-] Add the literals ''$clock(@x,q)$'' and ''$q\neq 0$'' to the body of each rule, where $x$ is the holding variable of the body.
\item[-] For each rule with an intensional relation $R$ of $\mathcal{P}_{DL}$ in its head, replace $R$ in the head with $tempR$ and add
\begin{eqnarray*}
R(\overrightarrow x)&:-&tempR(\overrightarrow x);clock(@x,0).\\
continue(@x)&:-&tempR(\overrightarrow x);\neg R(\overrightarrow x);clock (@x,0).\\
\uparrow inf(@y,x)&:-&tempR(\overrightarrow x);\neg R(\overrightarrow x);clock (@x,0); G(@x,y).
\end{eqnarray*}
in $\mathcal{P}_4$ where $x$ is the holding variable of both $R$ and $tempR$.

\item[-] Add

\begin{eqnarray*}
continue(@x)&:-&start(@x).\\
\uparrow inf(@y,x)&:-&start(@x);G(@x,y).\\
clock(@x,\kappa)&:-&start(@x).\\
clock(@x,p)&:-&clock(@x,q);q\geq 1; p=q-1;\neg stop(@x).\\
clock(@x,\kappa)&:-&clock(@x,0);\neg stop(@x).\\
\uparrow inf(@z,x)&:-&inf(@y,x);G(@y,z); x\neq z;clock(@x,q);q\geq \Delta.\\
continue(@x)&:-&inf(@x,y);clock(@x,q); q\neq 0.\\
continue(@x)&:-&continue(@x);clock(@x,q); q\neq 0.\\
stop(@x)&:-&\neg continue(@x);clock(@x,0).
\end{eqnarray*}
in $\mathcal{P}_4$.
\end{itemize}

\noindent{\bf Step 5: Inflationary result}

Input: $\mathcal{P}_4$. Output: $\mathcal{P}_{NL}$.

$\mathcal{P}_{NL}$ contains rules in $\mathcal{P}_4$ and the following rules:

\begin{itemize}
\item[-] For each relation $R$  in $\mathcal{P}_{4}$ except $start$, $clock$, $continue$, $inf$ and $stop$ but not in $\mathcal{P}_{DL}$, add
\[R(\dots@x\dots):-R(\dots@x\dots);clock(@x,q);q\neq 0.\]
in $\mathcal{P}_{NL}$.
\item[-] For each intensional relation R of $\mathcal{P}_{DL}$, add
\[R(\overrightarrow x):-R(\overrightarrow x).\]
in $\mathcal{P}_{NL}$. \hspace*{5.8cm} \qed
\end{itemize}

It is obvious that each  rule in a  program $\mathcal{P}_{NL}$ produced by the Rewriting Algorithm satisfies the localization restrictions, and can thus be computed effectively on one node. We can now state the main result of this section which shows that the global semantics of $\mathcal{P}_{DL}$ coincides with the distributed semantics of $\mathcal{P}_{NL}$.

\begin{theorem}\label{the-NL}
For a graph {\bf G}=\{V,G\}, a Datalog program $\mathcal{P}_{DL}$ and its rewritten $Netlog$ program $\mathcal{P}_{NL}$ produced by the Rewriting Algorithm, the computation of $\mathcal{P}_{NL}$ on {\bf G} terminates iff the computation of $\mathcal{P}_{DL}$ on {\bf G} terminates, and
$\mathcal{P}_{NL}(G)=\mathcal{P}_{DL}(G)$.
\end{theorem}

$\mathcal{P}_{NL}$ slows down the computation of $\mathcal{P}_{DL}$. During one stage ($\kappa$ rounds) of the computation of $\mathcal{P}_{NL}$, the clock turns from $\kappa$ to $0$, the sub-queries are evaluated and the sub-results are transmitted. At the end of each stage, the deduced facts for the intensional relations of $\mathcal{P}_{DL}$ are cumulated and all the sub-results are cleared. Hence, one such stage of $\mathcal{P}_{NL}$ is equivalent to one stage of $\mathcal{P}_{DL}$. For an intensional relation $R$ of $\mathcal{P}_{DL}$, $R(\overrightarrow c)\in I_{DLi}$ if and only if $R(\overrightarrow c)\in I_{NLi(\kappa+1)+1}$, $i\geq 0$, where $I_{DLi}$ and $I_{NLi}$ are the stages of respectively the fixpoints of $\mathcal{P}_{DL}$ and $\mathcal{P}_{NL}$.

The termination of the computation of $\mathcal{P}_{NL}$ is ensured by the predicate $stop$ as follows: the computation starts with a fact $start(a)$ on each node $a$, which triggers $clock(a,\kappa)$, $continue(a)$ and $inf(b,a)$ where $b$ is a neighbor of $a$. When the clock decreases from $\kappa$ to $0$, the evaluation of the sub-queries is done. The facts of an intensional relation $R$ of $\mathcal{P}_{DL}$ are stored in $tempR$. Meanwhile, $inf(v,a)$ is pushed to all the other nodes $v$ to inform that the computation on $a$ continues, so that $continue(v)$ is deduced. $continue(a)$ for one stage is maintained to the end of the stage. When the clock turns to zero, (i) the program checks if $continue(a)$ is true. If false, $stop(a)$ is deduced. Since $\neg stop(a)$ is a precondition for decreasing the clock and the clock is a precondition for deducing facts of all the other relations except $R$, so only the facts of $R$ are preserved along the stages. Thus the fixpoint is obtained and the computation terminates. Otherwise ($stop(a)$ is not deduced), the computation continues.  (ii) The programs compares facts of $tempR$ and $R$. If there are newly deduced facts, these facts are added into $R$. Meanwhile $continue(a)$ and $inf(b,a)$ are deduced for the next stage.

The proof of Theorem~\ref{the-NL} relies on the following Lemma and the fact that the Rewriting Algorithm produces only rules satisfying the localization restrictions.

\begin{lemma}\label{pro-NL}
For a graph {\bf G}=\{V,G\}, a Datalog program $\mathcal{P}_{DL}$ and its rewritten $Netlog$ program $\mathcal{P}_{NL}$ produced by the Rewriting Algorithm, the computation sequence $(I_{NLj})_{j\geq 0}$ for $\mathcal{P}_{NL}$ satisfies:
\begin{enumerate}
\item For each relation $R$ in $\mathcal{P}_{NL}$, $R(\overrightarrow c)\in I_{NLp}$ if\!f $R(\overrightarrow c)\in I_{NLp,c_1}$ and $R(\overrightarrow c)\notin I_{NLp,c'}$, where $c_1$ is the holding node of $R(\overrightarrow c)$ and $c'\neq c_1$.
\item $I_{NL0}=\{start(v)|v\in V\}$.
\item If $clock(a,c)\in I_{NLp}$, then $clock(v,c)\in I_{NLp}$ for all $v\in V$. If $stop(a)\in I_{NLp}$, then $stop(v)\in I_{NLp}$ for all $v\in V$.
\item If $stop(a)\in I_{NLs}$, then (i) $clock(a,\kappa)\in I_{NLs}$, (ii) for $q\in[1,s]$, $clock(a,\kappa-p)\in I_{NLq}$, $p\in[0,\kappa]$, if\!f $q=n(\kappa+1)+p+1$ and (iii) if $R(\overrightarrow c)\in I_{NLf}$ where $f>s+1$, then $R$ is an intensional relation of $\mathcal{P}_{DL}$. $Continue(a)\notin I_{NLn(k+)}$ for any $a\in V$ and any $p\geq s-(\kappa+1)$ if\!f $stop(a)\in I_{NLs}$. (iiii) $continue(a)\notin I_{NLs}$.
\item For each relation $R$ in $\mathcal{P}_{NL}$ but not in $\mathcal{P}_{DL}$, except the relations $start$, $clock$, $continue$, $inf$ and $stop$, (i) if $R(\overrightarrow c)\in I_{NLp}$, then $clock(a,\kappa)\notin I_{NLp}$, and (ii) if $p=n(\kappa+1)+q$, $q\in[2,\kappa+1]$, then $R(\overrightarrow c)\in I_{NLn(\kappa+1)+q'}$, $q'\in[q,\kappa+1]$.
\item For each intensional relation $R$ of $\mathcal{P}_{DL}$, if $R(\overrightarrow c)\in I_{NLp}$ then $R(\overrightarrow c)\in I_{NLp'}$ where $p'\geq p$. Assume $q=min\{p|R(\overrightarrow c)\in I_{NLp}\}$, then $clock(a,\kappa)\in I_{NLq}$.
\end{enumerate}
\end{lemma}

Now we prove Theorem~\ref{the-NL}.
\begin{proof}
Assume the computation sequence for $\mathcal{P}_{DL}$ is $(I_{DLi})_{i\geq 0}$ and for $\mathcal{P}_{NL}$ is $(I_{NLj})_{j\geq 0}$.
We prove for any intensional relation $Q$ of $\mathcal{P}_{DL}$, $Q(\overrightarrow c)\in I_{NLi(\kappa+1)+1}$ if\!f $Q(\overrightarrow c)\in I_{DLi}$.

\noindent {\sc Basis}: $i=0$, $I_{DL0}=\emptyset$ and $I_{NL1}=\{continue(a),inf(b,a),clock(a,\kappa)|a\in V,G(a,b)\}$.

\noindent{\sc Induction:} Suppose for $n\geq 0$, and each intensional relation $Q$ of $\mathcal{P}_{DL}$, \[Q(a_1,\dots,a_k)\in I_{DLn} \mbox{ if\!f }Q(a_1,\dots,a_k)\in I_{NLn(\kappa+1)+1}.\]

First we proof that for $n+1$, if $Q(b_1,\dots,b_k)\in I_{DLn+1}$, then $Q(b_1,\dots,b_k)\in I_{NL(n+1)(\kappa+1)+1}$.

If $Q(b_1,\dots,b_k)\in I_{DLn+1}$, then (i) $Q(b_1,\dots,b_k)\in I_{DLn}$ or (ii) $Q(b_1,\dots,b_k)$ is a newly deducted fact in $I_{DLn+1}$.

If $Q(b_1,\dots,b_k)\in I_{DLn}$ then $Q(b_1,\dots,b_k)\in I_{NLn(\kappa+1)+1}$ by the induction hypothesis, and $Q(b_1,\dots,b_k)\in I_{NLp}$ where $p\geq n(\kappa+1)+1$ by Lemma~\ref{pro-NL}.6, therefore $Q(b_1,\dots,b_k)\in I_{NL(n+1)(\kappa+1)+1}$.

Otherwise($Q(b_1,\dots,b_k)\notin I_{DLn}$), then there is one rule $r\in \mathcal{P}_{DL}$
\[\begin{array}{ll}r:Q(x_1,\dots,x_k):-R_1(\overrightarrow {y_1});\dots;R_m(\overrightarrow {y_m}); \neg R_{m+1}(\overrightarrow {y_{m+1}});\dots;\neg R_l(\overrightarrow {y_l}).\end{array}\]
and an instantiation $\sigma$ of the variables in $r$ such that $\sigma(x_i)=b_i$ for $i\in[1,k]$
\[
R_i(\sigma(\overrightarrow {y_i}))\left\{\begin{array}{ll}\in I_{DLn}\cup G, \mbox{ for }i\in[1,m]\\\notin I_{DLn}\cup G, \mbox{ for }i\in[m+1,l]\end{array}\right.
\]
and for some $e\in[1,m]$, $R_e(\sigma(\overrightarrow {y_e}))\notin I_{DLn-1}\cup G$.
By the induction hypothesis and Lemma~\ref{pro-NL}.1,
\[
R_i(\sigma(\overrightarrow {y_i}))\left\{\begin{array}{ll}\in I_{NLn(\kappa+1)+1,\sigma (hv_{R_i})}\cup G, \mbox{ for }i\in[1,m]\\\notin I_{NLn(\kappa+1)+1,\sigma (hv_{R_i})}\cup G, \mbox{ for }i\in[m+1,l]\end{array}\right.
\]
and $R_e(\sigma(\overrightarrow {y_e}))\notin I_{NL(n-1)(\kappa+1)+1,\sigma (hv_{R_i})}\cup G$. According to Lemma~\ref{pro-NL}.6

\[
R_i(\sigma(\overrightarrow {y_i}))\left\{\begin{array}{ll}\in I_{NLp,\sigma (hv_{R_i})}\cup G, \mbox{ for }i\in[1,m]\\\notin I_{NLp,\sigma (hv_{R_i})}\cup G, \mbox{ for }i\in[m+1,l]\end{array}\right.
\]
where $p\in [n(\kappa+1)+1,(n+1)(\kappa+1)]$ and $R_e(\sigma(\overrightarrow {y_e}))$ is newly deduced in $I_{NLn(\kappa+1)+1}$. So $continue(\sigma(hv_{R_e}))\in I_{NLn(\kappa+1)+1}$. By Lemma~\ref{pro-NL}.4, $stop(a)\notin I_{NLn(\kappa+1)+1}$ and $clock(a,\kappa-p)\in I_{NLn(\kappa+1)+1+p}$ for $p\in[0,\kappa]$ and for any $a\in V$.

Because
\[Q(@x_1,\dots,x_k):-tempQ(@x_1,\dots,x_k);clock(@x_1,0).\]
is in $\mathcal{P}_{NL}$, so if $tempQ(b_1,\dots,b_k)\in n(\kappa+1)+1+p$, $p\in[\kappa_r,\kappa]$, then $Q(b_1,\dots,b_k)\in I_{NL(n+1)(\kappa+1)+1}$ since $\kappa_r\leq \kappa$.

According to Rewriting Algorithm, $h_r=hv_Q$, $CN_r=\{h_r\}$ and
\begin{itemize}
\item
if all the holding variables of the literals in $body_r$ are the same with $h_r$ ($S'=S$), then
\[\begin{array}{ll}tempQ(@x_1,\dots,x_k):-R_1(\overrightarrow {y_1});\dots;R_m(\overrightarrow {y_m}); \neg R_{m+1}(\overrightarrow {y_{m+1}});\dots;\neg R_l(\overrightarrow {y_l}); clock(@x_1,q);q\neq 0.\end{array}\]and
\[\begin{array}{ll}tempQ(@x_1,\dots,x_k):-tempQ(@x_1,\dots,x_k); clock(@x_1,q);q\neq 0.\end{array}\]

are in $\mathcal{P}_{NL}$. $\kappa_r=1$. Therefore $tempQ(b_1,\dots,b_k)\in I_{NLn(\kappa+1)+1+p}$ for each $p\in[1,\kappa]$, and $Q(b_1,\dots,b_k)\in I_{NL(n+1)(\kappa+1)+1}$ by Lemma~\ref{pro-NL}.1 and ~\ref{pro-NL}.5.
\item
Otherwise, not all of the holding variables of the literals in $body_r$ are the same with $h_r$ ($S'\neq S$). Assume $hv_{R_1}=\dots=hv_{R_{w}}=hv_{R_{m+1}}=\dots=hv_{R_{m+u}}=h_r$. Then
\[\begin{array}{ll}tempQ(@x_1,\dots,x_k):-R_1(\overrightarrow {y_1});\dots;R_w(\overrightarrow {y_w}); \neg R_{m+1}(\overrightarrow {y_{m+1}});\dots;\neg R_{m+u}(\overrightarrow {y_{m+u}});\\
\hspace*{4cm}Q_1(\overrightarrow {z_1});\dots;Q_t(\overrightarrow {z_t}); clock(@x_1,q);q\neq 0.\end{array}\]and
\[\begin{array}{ll}tempQ(@x_1,\dots,x_k):-tempQ(@x_1,\dots,x_k); clock(@x_1,q);q\neq 0.\end{array}\]
are in $\mathcal{P}_{NL}$ where $Q_i(\overrightarrow{z_i})$ is in $head_{r_i}$ for $r_i\in \mathcal{P}_{NL}$. If for each $i\in[1,t]$, $Q_i(\overrightarrow {c_i})\in I_{NL}n(\kappa+1)+1+(\kappa_r-1)$, where $\overrightarrow{c_i}=\sigma(\overrightarrow {z_i})$, then $tempQ(b_1,\dots,b_k)\in I_{NL}n(\kappa+1)+1+\kappa_r$, then $tempQ(b_1,\dots,b_k)\in I_{NL}n(\kappa+1)+1+p$, $p\in[\kappa_r,\kappa]$. $r_i$ is as follows:

Literals $R_{w+1}(\overrightarrow {y_{w+1}})$, $\dots$, $R_{m}(\overrightarrow {y_{m}})$, $\neg R_{m+u+1}(\overrightarrow {y_{m+u+1}})$, $\dots$, $\neg R_{l}(\overrightarrow {y_{l}})$ are grouped into subsets $S''_1,\dots,S''_n$, such that he literals in different subsets have no common variables except the variable in $CN_r$ which is $x_1$.

For each $S''_i$,
\begin{itemize}
\item
if some of the holding variables of the literals in $S''_i$ are the neighbors of $h_r$, ($T_i\neq \emptyset$), then  $G(@hv_{\gamma_{iu}},h_r)$ where $hv_{\gamma_{iu}}$ is one of such variables, is added into $S''_i$. Then $h_{r_i}=hv_{\gamma_{iu}}$ and $CN_{r_i}=CN_r\cup \{h_{r_i}\}$. Literals in $S''_i$ along with ''$clock(@hv_{\gamma_{iu}},q)$'', ''$q\neq 0$'' constitute
$body_{r_i}$. ''$\uparrow Q_i(\overrightarrow {z_i})$'' constitute $head_{r_i}$ where $\overrightarrow {z_i}$ contains all the variables both in $body_{r_i}$ and in any of $R_{1}(\overrightarrow {y_1})$, $\dots$, $R_{w}(\overrightarrow {y_w})$, $\neg R_{m+1}(\overrightarrow {y_{m+1}})$, $\dots$,  $\neg R_{m+u}(\overrightarrow {y_{m+u}})$ or $head_r$, with $h_r$ as the holding node. If the evaluation for $r_i$ is finished, the result for the sub-query $Q_i$ gets to $\sigma(h_r)$ in the next round. $d_{r_i}=1$.

\item
Otherwise (non of the holding variables of the literals in $S''_i$ are the neighbors of $h_r$), ''$y=hv_{\gamma_{it}}$'' is added into $S''_i$ where $\gamma_{it}$ is one literal in $S''_i$ and $y$ does not occurs in $r$. $h_{r_i}=hv_{\gamma_{it}}$. $CN_{r_i}=CN_r\cup \{h_{r_i}\}$. Literals in $S''_i$ along with ''$clock(@y,q)$'', ''$y\neq 0$'' constitute
$body_{r_i}$. $head_{r_i}$ is ''$\uparrow Q_i(\overrightarrow {z_i})$'' where $\overrightarrow {z_i}$ contains all the variables both in $body_{r_i}$ and in any of $R_{1}(\overrightarrow {y_1})$, $\dots$, $R_{w}(\overrightarrow {y_w})$, $\neg R_{m+1}(\overrightarrow {y_{m+1}})$, $\dots$, $\neg R_{m+u}(\overrightarrow {y_{m+u}})$ or $head_r$, with $y$ as holding node.

Moreover, because the following rule is in $\mathcal{P}_{NL}$
\[\begin{array}{ll} \uparrow Q_i(@x\dots):-Q_i(@y\dots); G(@y,x);clock(@y,q);q\neq 0.\end{array}\]
therefore if the evaluation of $r_i$ is finished, then the result for the sub-query $Q_i$ is obtained locally in the next round and then is broadcast to every node in $\Delta$ rounds. $d_{r_i}=1+\Delta$.
\end{itemize}
For each $i\in[1,t]$, if $Q_i(\sigma(h_{r_i})\dots)\in I_{NLn(\kappa+1)+1+\kappa_{r_i}}$, then $Q_i(\overrightarrow {c_i})\in I_{NLn(\kappa+1)+1+(\kappa_{r_i}+d_{r_i})-1}$, and because $\kappa_r=max\{\kappa_{r_i}+d_{r_i}\}$, so $Q_i(\overrightarrow {c_i})\in I_{NLn(\kappa+1)+1+(p'-1)}$ for $p'\in[\kappa_{r_i}+d_{r_i},\kappa_r]$.

Each $r_i$ is then rewritten by $Rewrite$ function and the output rules are modified by the Rewriting Algorithm.
\end{itemize}
A set of rules $T_r\in\mathcal{P}_{NL}$ is obtained by applying the Rewriting Algorithm on $r$. For $r'\in T_r$ with some sub-queries \[\kappa_{r'}=max\{\kappa_{r'_i}+d_{r'_i}|r'\in T_r\mbox{ and }r'_i\mbox{ is a sub-query of }r'\},\] and for $r''\in T_r$ without sub-queries $\kappa_{r''}=1$. The answers to $r\in T_r$ is in $I_{NLn(\kappa+1)+1+\kappa_{r}}$. Therefore $Q_i(\sigma(h_{r_i})\dots)\in I_{NLn(\kappa+1)+1+\kappa_{r_i}}$ for each $i\in[1,t]$, so finally $Q(b_1,\dots,b_k)\in I_{NL(n+1)(k+1)+1}$.
%Meanwhile $\kappa_r$ is calculated by $Rewrite(r,h_r,CN_r)$ and is the maximum rounds for evaluation all the sub-queries $Q_i(\overrightarrow {z_i})$.

%%%%%%%%%%%
%%%%%%%%%%%%%
%%%%%%%%%%%%%%%

Then we proof that if $Q(b_1,\dots,b_k)\in I_{NL(n+1)(\kappa+1)+1}$ then $Q(b_1,\dots,b_k)\in I_{DLn+1}$ for $n+1$.

If $Q(b_1,\dots,b_k)\in I_{NL(n+1)(\kappa+1)+1}$, then

(i) $Q(b_1,\dots,b_k)\in I_{NL(n+1)(\kappa+1)}$ or

(iii) $tempQ(b_1,\dots,b_k)\in I_{NL(n+1)(\kappa+1)}$ and $clock(b_1,0)\in I_{NL(n+1)(\kappa+1)}$.

If $Q(b_1,\dots,b_k)\in I_{NL(n+1)(\kappa+1)}$, according to Lemma~\ref{pro-NL}.6, $Q(b_1,\dots,b_k)\in I_{NLn(\kappa+1)+1}$. By the induction hypothesis, $Q(b_1,\dots,b_k)\in I_{DLn}$, so $Q(b_1,\dots,b_k)\in I_{DLn+1}$.

Otherwise ($Q(b_1,\dots,b_k)\notin I_{NL(n+1)(\kappa+1)}$),
\[Q(@x\dots):-tempQ(@x\dots);clock(@x,0)\]
is in $\mathcal{P}_{NL}$, $tempQ(b_1,\dots,b_k)\in I_{NL(n+1)(\kappa+1)}$, $clock(b_1,0)\in I_{NL(n+1)(\kappa+1)}$. Therefore $stop(a)\notin I_{NL(n+1)(\kappa+1)+1}$, and by Lemma~\ref{pro-NL}.4, $clock(a,\kappa-p)\in I_{NLq}$, $p\in[0,\kappa]$ if\!f $q=n(\kappa+1)+1+p$.
A set of rules in $\mathcal{P}_{NL}$ of the following form, with $Q_i(\overrightarrow {x_i})$, $Q_{i1}(\overrightarrow{z_{i1}})$, $\dots$, $Q_{io}(\overrightarrow{z_{io}})$ as sub-queries, is used and only used for deducing $tempQ(b_1,\dots,b_k)$
\[\begin{array}{ll}(\uparrow) Q_i(\overrightarrow {x_i}):-R_{i1}(\overrightarrow {y_{i1}});\dots;R_{im}(\overrightarrow {y_{im}}); \neg R_{im+1}(\overrightarrow {y_{im+1}});\dots;\neg R_{il}(\overrightarrow {y_{il}});\\
\hspace*{2.5cm}Q_{i1}(\overrightarrow{z_{i1}});\dots;Q_{io}(\overrightarrow{z_{io}}); clock(@y,q);q\neq 0.\end{array}\]
and $tempQ(b_1,\dots,b_k)\in I_{NLp}$, $p\in[p',(n+1)(\kappa+1)]$ for some $p'\in[2,(n+1)(\kappa+1)]$.

According to Rewriting Algorithm, all of these rules are rewritten from a rule in $\mathcal{P}_{DL}$ with $Q(x_1,\dots,x_k)$ as the head and the literals $R_{t}(\overrightarrow {y_t})$ and $\neg R_{u}(\overrightarrow {y_u})$ occurring in these rules as the body.
Assume the rule is
\[\begin{array}{ll}r:Q(x_1,\dots,x_k):-R_1(\overrightarrow {y_1});\dots;R_m(\overrightarrow {y_m}); \neg R_{m+1}(\overrightarrow {y_{m+1}});\dots;\neg R_l(\overrightarrow {y_l}).\end{array}\]
By Lemma~\ref{pro-NL}.6,
\[
R_i(\sigma(\overrightarrow {y_i}))\left\{\begin{array}{ll}\in I_{NLn(\kappa+1)+1}\cup G, \mbox{ for }i\in[1,m]\\\notin I_{NLn(\kappa+1)+1}\cup G, \mbox{ for }i\in[m+1,l]\end{array}\right.
\]
where $\sigma(x_i)=b_i$ for $i\in[1,k]$, and for some $e\in[1,m]$, $R_e(\sigma(\overrightarrow {y_e}))\notin I_{NL(n-1)(\kappa+1)+1}\cup G$. By the induction hypothesis
\[
R_i(\sigma(\overrightarrow {y_i}))\left\{\begin{array}{ll}\in I_{DLn}\cup G, \mbox{ for }i\in[1,m]\\\notin I_{DLn}\cup G, \mbox{ for }i\in[m+1,l]\end{array}\right.
\]
and $R_e(\sigma(\overrightarrow {y_e}))\notin I_{DLn-1}\cup G$. So $Q(b_1,\dots,b_k)\in I_{DLn+1}$.

%%%%%%%
%%%%%%
%%%%%%%
Therefore for an intensional relation $Q$ of $\mathcal{P}_{DL}$, $Q(\overrightarrow c)\in I_{DLi}$ if and only if $Q(\overrightarrow c)\in I_{NLi(\kappa+1)+1}$.

We now proof that the computation of $\mathcal{P}_{NL}$ on {\bf G} terminates iff the computation of $\mathcal{P}_{DL}$ on {\bf G} terminates.

The computation of $\mathcal{P}_{DL}$ on {\bf G} terminates,\\
if\!f\\$(I_{DLj})_{j\geq 0}$ converges,\\
if\!f\\no new facts in any intensional relation of $\mathcal{P}_{DL}$ are deduced in $I_{DLi}$ for the minimal $i$,\\
if\!f\\no new facts in any intensional relation of $\mathcal{P}_{DL}$ are deduced in $I_{NLi(\kappa+1)+1}$ for the minimal $i$,\\
if\!f\\$continue(v)\notin I_{NL(i+1)(\kappa+1)}$ and $continue(v)\in I_{NLi(\kappa+1)}$,\\
if\!f\\ $stop(v)\in I_{NL(i+1)(\kappa+1)+1}$,\\
if\!f\\ $clock(v,c)\notin I_{NL(i+1)(\kappa+1)+2}$,\\
if\!f\\ only the facts in the intensional relations of $\mathcal{P}_{DL}$ are in $I_{NLp}$, $p> (i+1)(\kappa+1)+2$,\\
if\!f\\ $(I_{NLj})_{j\geq 0}$ converges,\\
if\!f\\ the computation of $\mathcal{P}_{NL}$ on {\bf G} terminates.

Therefore the computation of $\mathcal{P}_{NL}$ on {\bf G} terminates iff the computation of $\mathcal{P}_{DL}$ on {\bf G} terminates
and $\mathcal{P}_{NL}(G)=\mathcal{P}_{DL}(G)$.
\end{proof}

%%%%%%%%%%%%%%%%%%%%%%%%%%%%%%%%%%%%%%%%%%%%%%%
%%%%%%%%%%%%%%%%%%%%%%%%%%%%%%%%%%%%%%%%%%%%%%%
%%%%%%%%%%%%%%%%%%%%%%%%%%%%%

\bigskip

\section{Restriction to neighborhood}\label{sec-locfrag}

We next consider a restriction of FO and FP to bounded neighborhoods
of nodes which ensures that the distributed computation can be
performed with only a bounded number of messages per node.

Let $dist(x,y) \le k$ be the
first-order formula stating that the distance between $x$ and $y$ in the graph is
no more than $k$. Let ${\cal N}^k(x)=\{y|dist(x,y)\le k\}$ denote
the {\it $k$-neighborhood} of $x$.

Let $\varphi(x,\overrightarrow{y})$ be an FO formula with free variables
$x,\overrightarrow{y}$, then $\varphi^{(k)}(x, \overrightarrow{y})$ denotes the formula
with all the variables occurring in $\varphi$ relativized to the
$k$-neighborhood of $x$, that is each quantifier $\forall / \exists
z$ is replaced by $\forall / \exists z\in {\cal N}^k(x)$, and $y\in
{\cal N}^k(x)$ is added for each free variable $y$.

The local fragments of FO and FP can be defined as follows.

%\vspace*{-2.5mm}

\begin{definition}
\lfo\/ is the set of FO formulae of the form $\varphi^{(k)}(x,
\overrightarrow{y})$.
\end{definition}

%\vspace*{-2.5mm}

\noindent The local fragment of FP can be defined as fixpoint of
\lfo\/ formulae.

%\vspace*{-2.5mm}

\begin{definition}
\lfp\/ is the set of FP formulae of the form
$\mu(\varphi^{(k)}(T)(x, \overrightarrow{y}))$, where $\overrightarrow{y}=y_1\dots
y_{\ell}$ and $T$ is of arity $\ell+1$.
\end{definition}

%\vspace*{-2.5mm}

Consider again the examples of Section~\ref{sec-ql}.
It is easy to verify that the formula $\mu (\varphi(T)(x,h,d))$ defining the OLSR like
table-based routing is not in \lfp.
On the other hand, the formula $\mu(\varphi(ST)(x,y))$  defining the spanning tree is in \lfp, as well as
the  formulae
$\mu(\varphi(RouteReq)(x,y,d))$ and $\mu(\varphi(NextHop)(x,y,d))$ defining
the AODV like On-Demand Routing.

%\vspace*{-4mm}
\subsection{Distributed complexity}

We now show that the distributed computation of the local fragments,
$FO_{loc}$ and $FP_{loc}$, can be done very efficiently.
We assume that the nodes are equipped with ports for each of their
neighbors. The ports allow to bound the message size to a constant
independent of the network size. The proof relies as previously on
specific query engines for $FO_{loc}$ and $FP_{loc}$. The query
engine for $FO_{loc}$ works both in synchronous and asynchronous
systems.

\medskip
\noindent \textbf{Query engine for \lfo\/}
($\mathcal{QE}_{FO_{loc}}$) The requesting node broadcasts the
\lfo\/ formula $\varphi^{(k)}(x,\overrightarrow{y})$. For each node $a$, when
it receives the query $\varphi^{(k)}(x,\overrightarrow{y})$, it collects the
topology information of its $k$-neighborhood by sending messages of
$O(1)$ size, then evaluates $\varphi^{(k)}(a,\overrightarrow{y})$ (where $x$ is
instantiated by $a$) by in-node computation. Since all nodes collect
their $k-$neighborhood topology information concurrently, these
computations may interfere with each other. To avoid the interferences
between concurrent local computations of different nodes, the traces
of traversed ports are incorporated in all messages.

Each node collects the topology
of its $k$-neighborhood as follows.

%\vspace*{-5mm}
\begin{itemize}
\item For each node $a$, when it receives the
query $\varphi^{(k)}(x,\overrightarrow{y})$, it sends
a message (``collect'', $k$, $j$) to its neighbor though port $j$, and waits for replies.
\item Upon reception of a  message (``collect'', $i$, $j_1...j_{2(k-i)+1}$) by port
$j^\prime$, $a$ adds $j_1...j_{2(k-i)+1}j^\prime$ into a table $tracelist_a$, and
\begin{itemize}
\item if $i > 0$, $a$ sends
on each port $j^{\prime\prime}$ s.t. $j^{\prime\prime} \neq
j^\prime$ the message (``collect'', $i-1$, $j_1 ... j_{2(k-i)+1}
j^\prime j^{\prime\prime}$), and waits for replies;
\item  otherwise($i=0$),
$a$ sends  on port $j^\prime$ the message (``reply'', $j_1j_2 ...
j_{2k+1}$, $j^\prime$, $tracelist_a$).
\end{itemize}
\item Upon reception of a message (''reply'', $j_1\dots j_{2r+1}$, $j_{2r+2}$ $\dots$ $ j_{2k+2}$,
$tracelist'_1\dots tracelist'_{k-r+1}$) on port $j_{2r+1}$, and
replies from all the other ports have been received
\begin{itemize}
\item  if $r=0$, for $1\leq s\leq k+1$, $a$ stores in the local memory ($j_1 .... j_{2s} $, $tracelist^\prime_{s}$);
\item otherwise $a$ sends on port $j_{2r}$ a message (``reply'', $j_1 \dots j_{2r-1}$, $j_{2r}\dots j_{2k+2}$, $tracelist_a$ $tracelist'_1$$\dots$ $tracelist'_{k-r+1}$).
\end{itemize}

\item After receiving replies from all ports, $a$ computes the topology of the
$k$-neighborhood of $a$ by utilizing the stored tuples $(j_1 ...
j_{2r}, tracelist)$as follows:

Let \[\begin{array}{l}\mathcal{T}^k(a):=\{j_1 ... j_{2r} |
(j_1 ... j_{2r},tracelist) \mbox{ is stored in local memory of }a\}.\end{array}\]

Define an equivalence relation $\approx$ on $\mathcal{T}^k(a)$ as
follows: let $j_1 ... j_{2r},j^\prime_1 ... j^\prime_{2s} \in
\mathcal{T}^k(a)$, then $j_1 ... j_{2r} \approx j^\prime_1 ...
j^\prime_{2s}$ if and only if $\exists$ $(j_1 ... j_{2r},
tracelist), (j^\prime_1 ... j^\prime_{2s}, tracelist^\prime)$
s.t. $j_1 ... j_{2r} \in tracelist^\prime$
or $j^\prime_1 ... j^\prime_{2s} \in tracelist$.

The vertex set of the $k$-neighborhood of $a$ is \[\left\{j_1 ...
j_{2r} | j_1 ... j_{2r} \in \mathcal{T}^k(a), r \le k\right\} /
\approx,\] namely equivalence classes $[j_1\cdots j_{2r}]$ of
$\approx$ on elements $j_1 ... j_{2r}$ ($r \le k$) of
$\mathcal{T}^k(a)$.

Let $[j_1 ... j_{2r}],[j^\prime_1 ... j^\prime_{2s}]$ be two
vertices of the $k$-neighbor\-hood of $a$, then there is an edge
between $[j_1 ... j_{2r}]$ and $[j^\prime_1 ... j^\prime_{2s}]$ if
and only if there is $j^\ast_1 j^\ast_2...
j^\ast_{2t+1}j^\ast_{2t+2} \in \mathcal{T}^k(a)$ such that $j^\ast_1
... j^\ast_{2t} \approx j_1 ... j_{2r}$ and $j^\ast_1 ...
j^\ast_{2t+2} \approx j^\prime_1 ... j^\prime_{2s}$. \qed

\end{itemize}

%\vspace*{4mm}

We can now state our main result for \lfo.

\begin{theorem}\label{the-local-fo}
Let ${\bf G}=(V,G)$ be a network with $n$ nodes and diameter $\Delta$.
\lfo\/ formulae $\varphi^{(k)}(x,\overrightarrow{y})$ can be evaluated on
$G$ with the following complexity upper bounds:

\medskip
\begin{tabular}{cccc}
\hs IN-TIME/ROUND   &  DIST-TIME    & MSG-SIZE &  $\#$MSG/NODE \\
  $O(1)$ & $O(\Delta)$ & $O(1)$ & $O(1)$
\end{tabular}
\end{theorem}

Note that the distributed time $O(\Delta)$ comes from the initial
broadcasting of the formula. The computation itself is fully local,
and can be done in $O(1)$ distributed time. In the case of an asynchronous
system, DIST-TIME is bounded by $O(n)$.

We now consider  \lfp\/ which admits the same complexity bounds as
$FO_{loc}$ except for the distributed time. We first assume that the
system is synchronous, and discuss the asynchronous system later.

\noindent
\textbf{Query engine for \lfp} ($\mathcal{QE}_{FP_{loc}}$)\\
\noindent \textbf{Request flooding} The requesting node sets a clock
$\sigma$ of value $\Delta$ and broadcasts the message
$(\mu(\varphi^{(k)}(T)(x,\overrightarrow{y})),\Delta-1)$ to its neighbors. For
each node $a$, if it receives message
$(\mu(\varphi^{(k)}(T)(x,\overrightarrow{y})),c)$ and it haven't set the clock $\sigma$ before, then it sets a clock
$\sigma$ of value $c$, and if $c>0$, it broadcasts message
$(\mu(\varphi^{(k)}(T)(x,\overrightarrow{y})),c-1)$ to all its neighbors.

\noindent \textbf{Topology collection}
When the clock $\sigma$
expires, each node $a$ sets a clock $\sigma^\prime$ of value $4k$
and starts collecting all the topology information in its
$2k$-neighborhoods by sending messages and tracing the traversed
ports (like for Theorem~\ref{the-local-fo}). Now each
node $a$ gets a $2k$-local name for each
$a^\prime$ in its $k$-neighborhood, which is the set of traces from
$a$ to $a^\prime$ of length no more than $2k$, denote
this $2k$-local name of $a^\prime$ at $a$ by $Name^{2k}_{a}(a^\prime)$.

\noindent \textbf{Fixpoint Computation}
In each node $a$, there is a
local table to store the tuples $(a, \overrightarrow{b})$ in $T$, which uses the
$k$-local names $Name^k_{a}(a^\prime)$ of $a^\prime$.

When the clock $\sigma^\prime$ expires, each node $a$ sets a clock
$\tau=3k$ and starts evaluating the FO formula
$\varphi^{(k)}(T)(a,\overrightarrow{y})$ (where $x$ is instantiated by
$Name^{2k}_{a}(a)$, the $2k$-local name of $a$ at $a$).  Node
$a$ evaluates $\varphi^{(k)}(T)(a,\overrightarrow{y})$ by instantiating all the
(free or bounded) variables in $\varphi^{(k)}(T)(a, \overrightarrow{y})$ by its
$2k$-local names $Name^{2k}_{a}(a^\prime)$ for nodes in its
$k$-neighborhood and considering all the possible instantiations one
by one.

Suppose $a$ instantiates $(x, \overrightarrow{y})$ by $(a,\overrightarrow{b})$ and also
instantiates all the bounded variables, then a variable-free formula
$\psi$ is obtained. Since there may be atomic formulae $T(a^\prime,
\overrightarrow{b^\prime})$, $a$ should
send the query $?BT(a^\prime, \overrightarrow{b^\prime})$ to $a^\prime$,
then $a^\prime$ should check whether $T(a^\prime,
\overrightarrow{b^\prime})$ holds or not and send the answer to $a$. It works since from $Name^{2k}_{a}(b^\prime_i)$,
the $2k$-local names of $b^\prime_i$ at $a$, $a^\prime$ can get
$Name^{k}_{a^\prime}(b^\prime_i)$, the $k$-local names of $b^\prime_i$ at $a^\prime$.

During the above evaluation of $\varphi^{(k)}(T)(a,\overrightarrow{y})$, if a
new tuple $(a,\overrightarrow{b})$ satisfying $\varphi^{(k)}(T)(x, \overrightarrow{y})$ is
obtained, $a$ stores it in a temporary buffer (the local table
for $T$ will be updated later) by using the $k$-local names of $a$ and $\overrightarrow{b}$ at $a$, and sends messages to inform other nodes in
its $k$-neighborhood that new facts are produced.

For each node $a$, when the clock $\tau$ expires, it sets the value
of $\tau$ by $3k$ again; if some new tuples are produced, $a$
updates the local table for $T$, and empty the temporary buffer; if some new
tuples are produced or some informing messages are received, it
evaluates $\varphi^{(k)}(T)( a, \overrightarrow{y})$ again. \qed
\medskip

%\vspace*{-2mm}

\begin{theorem}\label{the-local-FP}
Let ${\bf G}=(V,G)$ be a network with $n$ nodes and diameter $\Delta$.
\lfp\/ formulae $\mu(\varphi^{(k)}(T)(x, \overrightarrow{y}))$ can be
evaluated on $G$ with the following complexity upper bounds

\medskip
\begin{tabular}{cccc}
\hs IN-TIME/ROUND   &  DIST-TIME    & MSG-SIZE &  $\#$MSG/NODE \\
$O(1)$ & $O(n)$ & $O(1)$ & $O(1)$
\end{tabular}
\end{theorem}
\begin{proof}
It is easy to see that messages sent during the computation of
$\mathcal{QE}_{FP_{loc}}$ are of size $O(1)$.

Before the clock $\sigma$ expires, it is evident that each node
sends only $O(1)$ messages of the format
$(\mu(\varphi(T)(x,\overrightarrow{y})), c)$.

Then each node sets the clock $\sigma^\prime$ and collects topology
information of its $2k$-neighborhood, since the degree of nodes is
bounded and in the $2k$-neighborhood of
$a$ there are only $O(1)$ nodes, each node sends only $O(1)$ messages as well.

After the clock $\sigma^\prime$ expires, each node $a$ sets the
clock $\tau$ and starts evaluating $\varphi^{(k)}(T)(a, \overrightarrow{y})$.
During each period $3k$ of $\tau$, node $a$ considers all the
possible instantiations of the (free or bounded) variables in
$\varphi^{(k)}(T)(a, \overrightarrow{y})$ one by one and evaluate the
instantiated formula. During each such period, since the total
number of different instantiations are $O(1)$ and only $O(1)$
messages are sent during the evaluation of each such instantiated
formula $\varphi^{(k)}(T)(a, \overrightarrow{b})$, the total number of messages
sent by $a$ is $O(1)$.

Moreover, after the clock $\sigma^\prime$ expires and before the
distributed computation terminates, each node $a$ only sends $O(1)$
messages: $a$ only be able to receive informing messages from nodes
in its $k$-neighborhood, the total number of tuples $(a,\overrightarrow{b})$
produced on nodes in the $k$-neighborhood of $a$ is $O(1)$, so the
total number of informing messages received by $a$ is $O(1)$,
consequently $a$ evaluates $\varphi^{(k)}(x,\overrightarrow{y})$ at most
$O(1)$ times, thus the total number of messages sent by $a$ is
$O(1)$.

After the clock $\sigma^\prime$ expires, during each period $3k$ of
clock $\tau$, there should be at least one informing message sent by
some node, which means at least one new tuple in $T$ is produced.
Since there are at most $O(n)$ number of tuples in $T$, the total
distributed time for the evaluation of
$\mu(\varphi^{(k)}(T)(x,\overrightarrow{y}))$ is $O(n)$.
\end{proof}

For asynchronous systems, a spanning tree rooted at the requesting
node can be used to evaluate $FP_{loc}$, and the
complexity bounds $\textrm{DIST-TIME}$ and $\#\textrm{MSG/NODE}$
become respectively $O(n^2)$ and $O(n)$.

\subsection{Networks with no global identifiers}

%\vspace*{-3mm}

The query engines $\mathcal{QE}_{FO_{loc}}$ and
$\mathcal{QE}_{FP_{loc}}$ evaluate \lfo\/ and \lfp\/ queries by
using only local names in the bounded neighborhoods of nodes, which
suggests that for the evaluation of the local fragments of FO and
FP, unique global identifiers for nodes are unnecessary. In this section,
we show that this is essentially the case, and consider their
evaluation on networks with identifiers which are only locally
consistent and on anonymous networks with ports.

%\vspace*{-2mm}

\begin{definition}
A network ${\bf G}=(V,G, L)$ with a labeling function $L: V \rightarrow C$
assigning identifiers to nodes, is \textbf{$k$-locally consistent}
if for each node $a \in V$, for any $b_1,b_2 \in \mathcal{N}^k(a)$,
$L(b_1) \neq L(b_2)$.
\end{definition}

%\vspace*{-2mm}

Ports have been used to construct local names in the previous
sub-section. They are not needed to evaluate \lfo\/ and \lfp\/ on
locally-consistent networks since these networks have locally unique
identifiers for nodes.

\begin{theorem}\label{the-local-fo-local-con-anony}
A \lfo\/ formula $\varphi^{(k)}(x, \overrightarrow{y})$ can be evaluated on
$k$-locally consistent networks with the following complexity upper bounds:

\medskip
\begin{tabular}{cccc}
\hs IN-TIME/ROUND   &  DIST-TIME    & MSG-SIZE &  $\#$MSG/NODE \\
  $O(1)$ & $O(\Delta)$ & $O(1)$ & $O(1)$
\end{tabular}
\end{theorem}

%\vspace*{-2mm}

%\vspace*{-2mm}

\begin{theorem}\label{the-local-fp-local-con-anony}
A \lfp\/ formula $\mu(\varphi^{(k)}(T)(x, \overrightarrow{y}))$ can be
evaluated on $k$-locally consistent networks with the following complexity upper bounds:

\medskip
\begin{tabular}{cccc}
\hs IN-TIME/ROUND   &  DIST-TIME    & MSG-SIZE &  $\#$MSG/NODE \\
  $O(1)$ & $O(n)$ & $O(1)$ & $O(1)$
\end{tabular}
\end{theorem}

Local fragments of FO and FP can also be evaluated with the same
complexity bounds on anonymous networks with ports since local names
can be obtained by tracing the traversed ports of messages.

Note that in general, FO and FP queries cannot be
evaluated over locally consistent or anonymous networks.

%\vspace*{-5mm}

%%%%%%%%%%%%%%%%%%%%%%%

\section{Conclusion}\label{sec-conclusion}

Fixpoint logic expresses at a global level and in a declarative way the interesting functionalities of distributed systems. We have proved that fixpoint formulae over graphs admit reasonable distributed complexity upper-bounds.

Moreover, we showed how global formulae can be translated into rule programs
describing the behavior of the nodes of the network and computing the same result.
The examples given in the paper have been implemented on the Netquest system which supports the Netlog language.
Finally, we proved the potential of restricted fragments of fixpoint logic to local neighborhood, that are still very expressive, but admit much tighter distributed complexity upper-bounds with bounded number of messages of bounded size, independent of the size of the network.

These results show how classical logical formalisms can help designing high level programming abstractions for distributed systems that allows to state the desired global result, without specifying its computation mode. We plan to pursue this investigation in the following directions. (i) Investigate the distributed complexity of other logical formalisms such as monadic Second Order Logic, which is very expressive on graphs. (ii) Study the optimization of the translation from fixpoint logic to Netlog, to obtain efficient programs. (iii) Extend these results to other distributed computing models.

%\medskip
%\noindent
\section{Acknowledgments}
The authors  thank Huimin Lin for fruitful discussions.

%\vspace*{-3mm}

%\begin{spacing}{0.6}
\bibliographystyle{abbrv}
\bibliography{bibliolocal}

\begin{thebibliography}{10}

\bibitem{AbiteboulAHM05}
S.~Abiteboul, Z.~Abrams, S.~Haar, and T.~Milo.
\newblock Diagnosis of asynchronous discrete event systems: datalog to the
  rescue!
\newblock In {\em Proceedings of the Twenty-fourth ACM SIGACT-SIGMOD- SIGART
  Symposium on Principles of Database Systems, Baltimore, Maryland, USA}, 2005.

\bibitem{AbiteboulHV95}
S.~Abiteboul, R.~Hull, and V.~Vianu.
\newblock {\em Foundations of Databases.}
\newblock Addison-Wesley, 1995.

\bibitem{AlonsoKSWW03}
G.~Alonso, E.~Kranakis, C.~Sawchuk, R.~Wattenhofer, and P.~Widmayer.
\newblock Probabilistic protocols for node discovery in ad hoc multi-channel
  broadcast networks.
\newblock In {\em Ad-Hoc, Mobile, and Wireless Networks, Second International
  Conference, ADHOC-NOW}, 2003.

\bibitem{AttiyaW04}
H.~Attiya and J.~Welch.
\newblock {\em Distributed Computing: Fundamentals, Simulations and Advanced
  Topics}.
\newblock Wiley-Interscience, 2004.

\bibitem{BejeranoBGR03}
Y.~Bejerano, Y.~Breitbart, M.~N. Garofalakis, and R.~Rastogi.
\newblock Physical topology discovery for large multi-subnet networks.
\newblock In {\em INFOCOM}, 2003.

\bibitem{BejeranoBORS05}
Y.~Bejerano, Y.~Breitbart, A.~Orda, R.~Rastogi, and A.~Sprintson.
\newblock Algorithms for computing qos paths with restoration.
\newblock {\em IEEE/ACM Trans. Netw.}, 13(3), 2005.

\bibitem{EbbinghausFlum99}
H.~Ebbinghaus and J.~Flum.
\newblock {\em Finite model theory}.
\newblock Springer-Verlag, Berlin, 1999.

\bibitem{fagin74}
R.~Fagin.
\newblock Generalized first-order spectra and polynomial-time recognizable
  sets.
\newblock In {\em Complexity of computation (Proc. SIAM-AMS Sympos. Appl.Math.,
  New York, 1973)}, pages 43--73. SIAM--AMS Proc., Vol. VII. Amer. Math. Soc.,
  Providence, R.I., 1974.

\bibitem{FungSG02}
W.~F. Fung, D.~Sun, and J.~Gehrke.
\newblock Cougar: the network is the database.
\newblock In {\em SIGMOD Conference}, 2002.

\bibitem{GrumbachLQ07}
S.~Grumbach, J.~Lu, and W.~Qu.
\newblock Self-organization of wireless networks through declarative local
  communication.
\newblock In {\em OTM, On the Move Conference, MONET Workshop}, volume LNCS
  4805, pages 497--506, 2007.

\bibitem{Immerman89}
N.~Immerman.
\newblock Expressibility and parallel complexity.
\newblock {\em SIAM J. Comput.}, 18(3):625--638, 1989.

\bibitem{LooCHMRS05}
B.~T. Loo, T.~Condie, J.~M. Hellerstein, P.~Maniatis, T.~Roscoe, and I.~Stoica.
\newblock Implementing declarative overlays.
\newblock In {\em 20th ACM SOSP Symposium on Operating Systems Principles},
  2005.

\bibitem{LooHSR05}
B.~T. Loo, J.~M. Hellerstein, I.~Stoica, and R.~Ramakrishnan.
\newblock Declarative routing: extensible routing with declarative queries.
\newblock In {\em ACM SIGCOMM Conference on Applications, Technologies,
  Architectures, and Protocols for Computer Communications}, 2005.

\bibitem{MaddenFHH05}
S.~Madden, M.~J. Franklin, J.~M. Hellerstein, and W.~Hong.
\newblock Tinydb: an acquisitional query processing system for sensor networks.
\newblock {\em ACM Trans. Database Syst.}, 30(1), 2005.

\bibitem{WiSeNts06}
P.~Marron and D.~{Minder et al.}
\newblock Embedded wisents research roadmap.
\newblock Technical report, Embedded WiSeNts Consortium, 2006.

\bibitem{RG95}
R.~Ramakrishnan and J.~Gehrke.
\newblock {\em Database Management Systems.}
\newblock McGraw-Hill, 2003.

\bibitem{RD95}
R.~Ramakrishnan and J.~Ullman.
\newblock A survey of deductive database systems.
\newblock {\em Journal of Logic Programming.}, 23(2), 1995.

\bibitem{ReissH06}
F.~Reiss and J.~M. Hellerstein.
\newblock Declarative network monitoring with an underprovisioned query
  processor.
\newblock In {\em ICDE}, 2006.

\end{thebibliography}
%\end{spacing}
%\balancecolumns

\end{document}